\documentclass[10pt,a4paper]{article}

\usepackage[a4paper, margin=0.75in]{geometry} 

\usepackage[utf8]{inputenc}
\usepackage[T1]{fontenc}

\usepackage{algorithmic}
\usepackage{amsfonts}
\usepackage{amsmath}
\usepackage{amssymb} 
\usepackage{amsthm}
\usepackage{anyfontsize}
\usepackage{booktabs} 
\usepackage{caption} 
\usepackage{color}
\usepackage{comment}
\usepackage{courier} 
\usepackage{dsfont}
\usepackage{enumitem}
\usepackage{float}
\usepackage{graphicx} 
\usepackage{harpoon} 
\usepackage{latexsym}
\usepackage{mathtools}
\usepackage{qtree}
\usepackage{slashed} 
\usepackage{subcaption}
\usepackage{tikz-cd}
\usepackage{url}
\usepackage{verbatim}
\usepackage{xcolor}
\usepackage{xurl}
\usepackage{breqn}

\usepackage{hyperref}
\definecolor{darkblue}{rgb}{0, 0, 0.5}
\hypersetup{colorlinks=true,citecolor=darkblue, linkcolor=darkblue, urlcolor=darkblue}

\usepackage[style=numeric, 
natbib=true, 
backend=biber, 
giveninits=true, 
uniquename=init,
isbn=false,
doi=false,
eprint=false,
maxbibnames=3, 
maxcitenames=2
]{biblatex}
\DeclareMathOperator{\Tr}{{Tr}}

\DeclareMathOperator{\Lin}{{Lin}}
\DeclareMathOperator{\LinSpan}{{LinSpan}}

\DeclareMathOperator{\End}{{End}}

\DeclareMathOperator{\Fock}{{\Gamma}}

\DeclareMathOperator{\Range}{{\textsc{Range}}}

\DeclareMathOperator{\Dom}{{Dom}}

\newcommand{\tsr}{\otimes}

\newcommand{\closure}[1]{\overline{#1}}

\newcommand{\ket}[1]{{|#1\rangle}}

\newcommand{\inner}[1]{{\langle #1\rangle}}
\newcommand{\ip}[1]{{\langle #1\rangle}}

\newcommand{\NN}{{\mathds{N}}}

\newcommand{\laplace}{\mathop{}\!\mathbin\bigtriangleup}

\newcommand{\wtilde}[1]{\widetilde{#1}}

\newcommand{\grad}{\nabla}

\newcommand{\Div}{\text{div}}

\newcommand{\footurl}[1]{{\footnote{\href[#1]{#1}}}}

\newcommand{\one}{{\textbf{1}}}

\newcommand{\norm}[1]{\left\lVert#1\right\rVert}
\newcommand{\vertiii}[1]{{\left\vert\kern-0.25ex\left\vert\kern-0.25ex\left\vert #1 \right\vert\kern-0.25ex\right\vert\kern-0.25ex\right\vert}}

\newtheorem{thm}{Theorem}[section]
\newtheorem{definition}[thm]{Definition}
\newtheorem{prop}[thm]{Proposition}
\newtheorem{remark}[thm]{Remark}
\newtheorem{claim}[thm]{Claim}

\newtheorem{estimate}[thm]{Estimate}
\newtheorem{theorem}[thm]{Theorem}

\newtheorem{lemma}[thm]{Lemma}

\newtheorem{corollary}[thm]{Corollary}

\theoremstyle{definition}

\newcommand{\E}{\mathrm{E}}

\newcommand{\CStar}{\emph{C}\textsuperscript{*}}

\newcommand{\eps}{\epsilon}

\newcommand{\Fc}{{\mathcal{F}}}

\newcommand{\Ec}{{\mathcal{E}}}

\newcommand{\Uc}{{\mathcal{U}}}

\newcommand{\Ac}{{\mathcal{A}}}
\newcommand{\Bc}{{\mathcal{B}}}

\newcommand{\Wc}{{\mathcal{W}}}
\newcommand{\Vc}{{\mathcal{V}}}
\newcommand{\Lc}{{\mathcal{L}}}
\newcommand{\Kc}{{\mathcal{K}}}

\newcommand{\C}{{\mathcal{C}}}
\newcommand{\B}{{\mathcal{B}}}

\newcommand{\Hs}{{\textsf{H}}}

\newcommand{\Vs}{{\textsf{V}}}

\newcommand{\CC}{\mathbb{C}}

\newcommand{\KK}{\mathbb{K}}

\newcommand{\RR}{\mathds{R}}

\newcommand{\RHS}{\emph{r.h.s.}}

\newcommand{\txt}[1]{\text{#1}}

\newcommand{\hsep}{\begin{center}
{\rule{0.5\linewidth}{0.5pt}}
\end{center}
}

\newcommand{\mat}{\textsc{Mat}}

\newcommand{\namecite}[1]{\citet{#1}}

\newcommand{\SwapOp}[1]{{\txt{Swap}_{#1}}}

\newcommand{\DenseIn}{{\subset_\txt{{dense}}}}

\newcommand{\Sect}{\Gamma}
\newcommand{\ExpHVec}{\textsc{E}}

\newcommand{\sectiontitlefont}[1]{#1}

\newcommand{\Rem}[1]{} 
\newcommand{\FRem}[1]{} 
\newcommand{\FRemHistorical}[1]{}

\newcommand{\FOnlineRef}[1]{\FRem{\href{#1}{online}}} 

\newcommand{\ToDo}[1]{{}}

\usepackage{endnotes}

\usepackage{imakeidx}
\makeindex

\begin{document}

\title{\textsc{Spectral action and heat kernel trace for Ricci flat manifolds from stochastic flow over second quantized $L^2$-differential forms. 
}}
\author{Sita Gakkhar\footnote{Sitanshu Gakkhar, \href{mailto:sgakkhar@caltech.edu}{email}. The first author would like to thank Debashish Goswami for helpful discussions on qsde's of Evans-Hudson type, and also Branimir {\'C}a{\'c}i{\'c} for discussions on Dirac operators}, Matilde Marcolli,\\}
\maketitle

\begin{abstract}
A quantum stochastic differential equation (qsde) on Fock space over $L^2$ differential $1$-forms is given from the small ``time'' flow of which the trace of the connection Laplacian heat kernel for the spinor endomorphism bundle can be computed over any compact Ricci-flat Riemannian manifold. The existence of the stochastic flow is established by adapting the construction from \cite{goswami_qsp}. When the manifold supports a parallel spinor -- Ricci-flatness is a required integrability condition for parallel spinors, the trace of Dirac Laplacian heat kernel of the spinor bundle can be recovered. For $4$-manifolds, this corresponds to the spectral action, and realizes Einstein-Hilbert action as a stochastic flow. 
\end{abstract}


\section{Introduction}

This article attempts to formalize the idea that Einstein-Hilbert action for a Riemannian manifold can be viewed as arising from random fluctuations acting on the spinor bundle. The model of random fluctuations is provided by a stochastic flow generated by the Dirac Laplacian. The connection between Dirac Laplacian, noncommutative geometry and gravity is well established (see \cite{tolksdorf_analysis_of_dirac_operators} and references therein for an account), the new contribution here is the probabilistic perspective. 

This realization of spectral action as a stochastic flow is suggested by explicit computations of spectral action which yield Brownian bridge integrals along with the observation that boson Fock space can be viewed as a Wiener space, with a preliminary exploration of the perspective put forward in \cite{gakkhar2022quantum} where covariant quantum diffusions on almost-commutative spectral triples are considered. This article treats canonical spectral triples, computing explicitly the structure matrix for the diffusion generated by the endomorphism connection Laplacian acting on the algebra of functions over any compact Riemannian manifold. In absence of a natural action with respect to which the generator is covariant, the standard constructions using Picard iterates (see \namecite{goswami_qsp}, and also \cite{belton2015algebraic}) care adapted to show the existence of the solution. From this flow, when there exists a parallel spinor, spectral action can be evaluated. 

\emph{Some remarks on notation.} By Riemannian $(M, g)$, we mean a Riemannian manifold $M$ with metric $g$. The connection on the tangent bundle of $M$, $TM$, is always the Levi-Civita connection unless specified otherwise. When clear from context, the same symbol is used of the connection $\grad$ on a Hermitian or Riemannian bundle $E$ and the dual connection on dual bundle $E^*$. After fixing a local orthonormal frame about any $p\in M$, $(X_i)_{i\in \dim M}$, $\grad_{X_i}$ will be used interchangeably with $\grad_i$. For local coordinates $(x_i)$ about any $p\in M$, $\partial_i$ will denote the coordinate vector fields $\frac{\partial}{\partial x_i}$. $[n]$ is the set $\{i\in \NN, i\leq n\}$ where $\NN$, with convention that $0\not\in \NN$.  $[n:m]$ denotes the set $\{n, n+1\dots m\}$. 
The linear span we mean finite linear span denoted by $\LinSpan(\Vs):=\{\sum_{i\in [k]}\alpha_ia_i: \alpha_i\in \KK, a_i\in \Vs\}$ where $\KK=\RR, \CC$ is clear from  context; if $\KK=\RR$ then the subscript is dropped. Throughout $\Fock(H)$ denotes the symmetric (boson) Fock space over the any space $H$, while $\ExpHVec(H)$ denotes the exponential vectors given by $\ExpHVec(v) = \oplus_{n=0}^\infty (n!)^{-1/2}v^{\tsr n}$ for $v\in H$. From section~\ref{section_growth_bounds} onwards, we make use of Einstein notation. All calculation prior to section~\ref{section_growth_bounds} are using at the center of normal coordinates where the metric is identity, for clarity, the explicit summations are used. 

\subsection{Organization and overview}

In the remainder of this section, we introduce the background on stochastic flows and spectral action, and delve into motivating ideas. In section~\ref{section_structure_matrix}, the structure matrix for connection Laplacian is computed following the standard prescription. The noise space turns out to be $L^2$ differential forms and the flow lives on the Fock space over the differential forms. The section concludes by writing the structure matrix and the qsde for the stochastic flow in the coordinate free quantum stochastic calculus notation; the relevant material from quantum stochastic differential equations an appendix reviewing quantum stochastic differential equations is included at the end (appendix \ref{section_mapvalued_qsdes}). Necessary bounds for controlling growth of Laplacian powers are established in section~\ref{section_growth_bounds}. The existence of flows for the derived qsde is established in section~\ref{section_picard_iterates_convergence} by providing estimates that can be plugged into the standard theory. 

\subsection{Spectral action and stochastic flows}

The canonical spectral triple for Riemannian $(M, g)$ which carries a spin structure is the data $(C^\infty(M), L^2(S), D_M)$, $L^2(S)$ being the Hilbert space of square integrable sections of a spinor bundle $S\to M$, and $D_M$ the Dirac operator associated to the lift of Levi-Civita connection to the spinor bundle\cite[pg~67]{suijlekom_ncg}. We will take $S$ to be any spinor bundle associated to $TM$ and $D_M$ to be the Atiyah-Singer operator $\slashed{D}$ (see \cite[ex~II.5.9]{lawson_spin}), then we have that $\slashed{D}^2=\laplace + \kappa/4$ where $\laplace$ is the connection Laplacian for the connection on $S$ and $\kappa$ is the scalar curvature\cite[Thm~II.8.8]{lawson_spin}. 

The bosonic spectral action is the linear funtional $S_b^M:=\Tr f(D_M/\Lambda)$ for a choice of test function $f$ which we take as $e^{-x^2}$ and $\Lambda$ a cutoff parameter\cites[\S~5.1]{ncg_standard_model}[\S~7.1]{suijlekom_ncg}, so that the parameter $t$ of the Dirac heat semigroup $e^{-tD_M^2}$ will be taken to satisfy $t=\Lambda^{-2}$. From the asymptotic expansion $\lim_{\Lambda\to \infty} S_b^M$ for a Riemannian spin $4$-manifold $M$, the Einstein-Hilbert action, $S_{EH}$, can be recovered\cites[\S~5.3]{ncg_standard_model}[\S~8.3]{suijlekom_ncg}. 

Since $D_M, D_M^2$ are self-adjoint on $L^2(S)$, suppose $(\phi_{i,j})$ is an basis of orthonormal eigensections with eigenvalues $\lambda_{i,j}^2$ for $D^2_M$ where $j\in [n_i]$ runs over the multiplicity $n_i$ of eigenvalue $\lambda^2_{i,j}:=\lambda^2_i$ 
Define the state $\Phi^z=  1/N(z)\sum_{i:\lambda_i \leq z } \sum_{j\in [n_i]} \phi_i^j$ with normalization $N(z):=\sqrt{\sum_{i:\lambda_i\leq z} n_i}$ then 
\begin{align*}
\lim_{z\to \infty}
\inner{\Phi^z, e^{-tD_M^2}\Phi^z}N(z) &= \sum_{i:\lambda_i\leq z}\sum_{j\in n_i}\inner{\phi_{i,j}, e^{-tD_M^2}\phi_{i,j}} =
\lim_{z\to \infty} \inner{\Phi^z, e^{-tD_M^2}\Phi^z} = \Tr e^{-tD_M^2}
\end{align*}
Therefore, the bosonic spectral action can be approximated by expectation of small time expectation of $e^{-t(D_M)^2}$ in state $\Phi^z$ for $z$ large: $$
S_{EH}\approx \lim_{\sqrt{t}\to 0}\Tr e^{-tD_M^2}  = \lim_{\sqrt{t}\to 0}\lim_{z\to \infty}
\inner{\Phi^z, e^{-tD_M^2}\Phi^z}N(z)  
$$
This motivates the interest in evaluating $\inner{\Phi^z, e^{-tD_M^2}\Phi^z}$. The approach we take is that of a quantum stochastic dilation associated to the heat semigroup, $e^{-tD^2_M}$, which is a quantum dynamical semigroup by \cite{gakkhar2022quantum}. 

The issue is complicated by the fact that for the existence of Evans-Husdon flow, the generator must annihilate identity, that is, the semigroup must be conservative. The Dirac laplacian acting on endomorphisms by composition does not satisfy this. We instead have to work with the endomorphism connection and the associated endomorphism connection laplacian and endomorphism Dirac laplacian and then extract the spinor bundle Dirac laplacian from it (see \cite{gakkhar2022quantum} more detailed discussion on this). Because of Ricci flatness, $\slashed{D}^2=\laplace$, the semigroup of interest will be $e^{-t\laplace^\End}$ for the endomorphism connection laplacian $\laplace^\End$. 

\begin{remark}If the spinor bundle could be replaced by the Clifford bundle, then the associated Dirac laplacian does generate a flow of Evans-Hudson type and the endomorphism trick is not needed. \end{remark} 

A quantum stochastic dilation of the Evans-Hudson type\footnote{see appendix \ref{section_mapvalued_qsdes} for a brief review of quantum stochastic processes and integrals, and \cite{goswami_qsp,parthasarathy_qsp} for detailed standard theory} for the quantum dynamical semigroup $e^{-t\Lc}$ on the \CStar-algebra $\Ac:=\C(M)\subset \End(L^2(S))$ is a family of $*$-homomorphisms $j_t:\Ac\to \Ac''\tsr \Bc(\Fock(\RR_+, k_0))$ where $k_0$ is a Hilbert space, called the noise (or multiplicity) Hilbert space that is constructed from the generator $-\Lc$,  such that the following holds \begin{itemize}
    \item There exists  an ultra-weak dense $*$-subalgebra $\Ac_0\subset \Ac$ such that the map-valued process $J_t:\Ac\tsr \ExpHVec(L^2(\RR_+, k_0))\to \Ac''\tsr \Fock(L^2(\RR_+, k_0))$ with $J_t(x\tsr \ExpHVec(f))u := j_t(x)(u\ExpHVec(f))$ for $x\in \Ac, u\in \Hs:=L^2(S)$, $f\in L^2(\RR_+, k_0)$, satisfying the qsde: \begin{align}\label{eq_evans_hudson_qsde}
    dJ_t = J_t\circ (a_\delta(dt) + a_\delta^\dagger(dt) + \Lambda_\sigma(dt) + I_\Lc(dt)), J_0 =\one
    \end{align}
    on $\Ac_0\tsr \ExpHVec(L^2(\RR_+,k_0))$ where $\delta:\Ac_0\to \Ac\tsr k_0, \sigma:\Ac_0\to \Ac\tsr \Bc(k_0)$ are linear maps, called the structure maps for the qsde,  derived from the generator $\Lc$ for $e^{-tD_M^2}$, $\Ac_0\subset \Dom(\Lc)$, and $a_\delta, a^\dagger_\delta, \Lambda_\sigma, I_\Lc$ the fundamental processes with respect to which stochastic integral is defined. 
        \item For all $u, v\in \Hs, x\in \Ac$,
        \begin{align}
    \inner{v\ExpHVec(0), j_t(x)\ExpHVec(0)} = \inner{v, e^{-t\Lc}u}
    \end{align}
\end{itemize}

Therefore, $\inner{\Phi^z, e^{-t\Lc}\Phi_z}$ is realized as operator algebraic expectation of $j_t$ with respect vacuum state $\ket{\Phi^z\ExpHVec(0)}\in L^2(S)\tsr \Fock(L^2(\RR_+, k_0))$, where $j_t$ is obtained from the flow $J_t$ for the Evans-Hudson qsde (equation~\ref{eq_evans_hudson_qsde}). Schemes for solving for Evans-Hudson qsde, for example, using Picard iterates, provide a way to algorithmically construct the flow. 

\subsection{Dirac Laplacian, endomorphism connection and parallel spinors}

To start we note the following sign conventions of the Laplacians. Primarily the signs are fixed so the Laplace-Beltrami operator has non-negative spectrum, and signs on all other Laplacians cascade from there. On Riemannian $(M, g)$,  $M$ compact, without boundary, $\Tr_g$ denotes the trace of a covariant tensor\FRem{\cite[pg~28]{lee_riemannian}} taken after identifying with a contravariant tensor via the metric $g$, $\Tr_g(h) := g^{ij}h_{ij}$. Note that trace on contravariant tensor, e.g. vector fields, is simply the sum. For $X\in \Sect(TM)$, $\Div(X) = \Tr(\grad X)$. The Laplace-Beltrami operator is taken as the operator with non-negative spectrum, that is, $-\Div(\grad) = -\Tr(\grad_{\cdot,\cdot})$, where $\grad_{\cdot,\cdot}$ is the second invariant derivative $\grad^2_{V, W}:=\grad_{V}\grad_{W} - \grad_{\grad_VW}:\Sect(E)\to \Sect(E)$. The connection Laplacian is $\grad^*\grad$ where $\grad^*$ is adjoint of the connection $\grad:\Sect(E)\to \Sect(E)\tsr T^*M$. Equivalently, $\grad^*\grad = -\Tr(\grad_{\cdot,\cdot})$. Further, $\laplace = -g^{ij}\grad_i\grad_j$.

Let $\grad$ be any connection on the vector bundle $E\to M$, $(M, g)$ a compact Riemannian manifold. The connection Laplacian at $p\in M$, $\laplace=\grad^*\grad$ in local coordinates\FRem{\cite[pg~66]{berline_heat}} $(e_i)$ is given by $\laplace=-(\sum_i \grad_i\grad_i - \grad_{\grad_i e_i}$). To evaluate $\laplace \phi$ at any $p\in M$ and $\phi\in \Sect(E)$, we will use Riemann  normal coordinates centered at $p$ so $\grad_i e_j$ vanish, yielding $\laplace \phi(p) = -\sum_i \grad_i\grad_i\phi(p)$.

The endomorphism connection $\grad^\End$ on the bundle $\End(E)=E\tsr E^*$ associated to a connection $\grad$ on the Hermitian vector bundle $E$ over the Riemannian manifold $M$ is such that for $X\in TM, \grad_X^\End = \grad_X\tsr 1 + 1\tsr \wtilde\grad_X$, where $\wtilde\grad$ is the dual connection on $E^*$. 
The endomorphism Laplacian is defined as usual: at $p\in M$ in Riemann normal coordinates centered at $p$ (denoting $\wtilde\grad^*\wtilde\grad, \wtilde\grad$ by   $\laplace, \grad$ again), $$
\laplace^{\End} = -\sum_i \grad_i^\End \grad_i^\End = -\left(\laplace\tsr 1 + 2\sum_i\grad_i\tsr \grad_i + 1\tsr \laplace\right)$$
Note that as $E\tsr E^*$ is balanced over $\C(M)$, the action of $\C(M)$ on $\End(E)$ can be written as $f\cdot 1_{\End} = \sum_i(f\cdot h_i)\tsr h_i^*$; this convention is used for all computation with Laplacian expressed in this tensor form. It's also very useful to note that in any local coordinates , $\grad^\End$ acts by commutator: if over chart $U$, the connection has potential $A$, $\grad = d + A$, then for a local orthonormal frame $(\mu_i)$ and dual frame $(\mu^j)$, $\grad^{\End}\sum_{ij}\sigma^i_j\mu_i\tsr \mu^j =
\sum_{ij}(d\sigma_j^i)\mu_i\tsr \mu^j +
\sum_{jk}[\sigma A - A\sigma]_{jk}\mu_k\tsr \mu^j$. In particular, since $\one_\End$ is given by the identity matrix locally, it follows (see \cite{gakkhar2022quantum}) that $\grad^\End(\one_\End)=0$. This implies that again in normal Riemann coordinates centered at $p$ yields that for any $f\in \C^\infty(M)$, \begin{align} \laplace^\End(f\one_\End)(p) &= -\sum_i\grad^\End_i\grad^\End_i (f\cdot \one_\End)(p) = \sum_i\grad_i^\End\left(\partial_i f\one_\End\right) = -\sum_i\partial_i\partial_if\cdot \one_\End(p)
\end{align}

\begin{prop}\label{prop_laplace_beltrami_ise_endomorphsim_Laplacian} For $f\in \C^\infty(M)$,   $\laplace^{\End(E)}(f\cdot \one_\End) = \laplace^M(f)\cdot\one_\End$.
\end{prop}
\begin{proof} 
Let $\Gamma^{k}_{ij}$ be the Christoffel symbols for Levi-Civita connection, then in local coordinates about $x\in M$, $\laplace^M(f) = -\sum_{ij}g^{ij}(x)(\partial_i\partial_j - \sum_k \Gamma^k_{ij}\partial_k)f$ (see, for instance, \cite[pg~66]{berline_heat}) and for the endomorphism Laplacian,   
\begin{align*}
\laplace^\End(f\cdot 1) &= -\sum_{ij}g^{ij}(x)(\grad^\End_i\grad^\End_j - \sum_k \Gamma^{k}_{ij}\grad^\End_k)(f\cdot\one_\End) \\
&= -\left(\sum_{ij}g^{ij}(x)(\partial_i\partial_j - \sum_k \Gamma^k_{ij}\partial_k)f\right)\cdot \one_\End =\laplace^M(f)\one
\end{align*}
where we used $\grad^\End(\one_\End)=0, \grad_X^\End(f)\cdot \one_\End = X(f)\cdot \one_\End$. 
\end{proof}

Now for any constant $c$, $\Tr(e^{-t(\laplace+c)})$ is just $e^{-ct}\Tr(e^{-t\laplace})$. If there exists a parallel section $\phi$ for $\grad$ (equivalently for the connection on the dual bundle), then with $H_\phi = \closure{\{s\tsr\phi:s\in \Gamma(E)\}}$, and appropriate normalization on $\phi$, \begin{align}\label{eq_trace_from_endomorphism_trace} \Tr(e^{-t\laplace}) = \Tr|_{H_\phi}(e^{-t\laplace^{\End}})\end{align} 
The parallel section, therefore, allows one to extract the heat kernel trace $\Tr(e^{-t\laplace})$ from the $\Tr(e^{-t\laplace^\End})$. 

In the setting of the canonical spectral triple, we specialize to the spinor bundle. The existence of a parallel spinor constrains the holonomy of Levi-Civita connection of the manifold. In particular, this implies for Riemannian $(M, g)$ that the Ricci tensor vanishes\cite[\S~6.3]{ofarrill_spin}, therefore, $D^2=\laplace$. 

\begin{remark}{A remark on the Lorentzian and Kahler analogs:} the existence of a parallel spinor constrains the holonomy\cite{ofarrill_spin}, so is a strict condition. But the analysis is expected to work for Dirac operators coming from $\text{Spin}^\CC$. The extension to the complex setting, especially Kahler manifolds, should also follow easily, and these provide an interesting class of examples. In the Lorentzian setting, there's a richer supply of parallel spinors, however, the essential difficulty there is that the spectrum of the Dirac operator is no longer discrete and the regularity requirements become unclear. The spectral action principle for Lorentzian scattering space obtained by \cite{dang2020complex} suggests the obvious question of a probabilistic interpretation in the Lorentzian setting as well. Since $\C(M)$ is a commutative algebra, so positivity is equivalent to complete positivity, therefore, by same ideas, $e^{-t\laplace^\End}\subset \C(M)$ is a completely positive semigroup, and the question is reasonable.    
\end{remark}

\begin{remark}{A remark on in absence of a parallel spinor:} The requirement of the parallel section can be worked around sometimes, for example, $S^1\times S^3$ where $S^1$ carries the flat connection. By \cite{gakkhar2022quantum}, there is a quantum stochastic flow for the Laplacian for the endomorphism connection associated to the homogeneous $H$-connection on the spinor bundle over $S^3$. Since $S^3$ is also symmetric the homogeneous connection and the Levi-Civita connection agree.  
However, there are no parallel spinors on $S^3$, but there do exist Killing spinors with Killing constant $\pm1/2$. In such a setting, it's possible to use the homogeneous space construction for quantum diffusion on the endomorphism bundle to get at the spectral action and the heat kernel, the idea being to modify the connection to make a Killing spinor parallel. Let $\laplace'_\pm$ be the connection Laplacians, with $k=\pm 1/2$, for the connections $\grad'_\pm$, acting by $\grad'_X \phi = \grad_X\phi \pm kX\cdot \phi$ where $\cdot$ is the Clifford action, and $\grad$ the Levi-Civita connection, then on $S^n$, the Wietzenbock identity $(D\pm k)^2 = \laplace'_{\pm} + \frac{1}{4}(n-1)^2$ holds (this is a calculation, that works more generally than $S^n$). 
Using that Killing spinors are parallel for $\laplace_\pm$, $\Tr(e^{-t(D\pm k)^2}) = e^{-t/4(n-1)^2} \Tr(e^{-t\laplace_{\pm}})$ can be computed.
It just needs to be checked that the flow exists for $\laplace_{\pm}$. For $S^1\times S^3$ examples, the Dirac operator and its square can be explicitly computed by specializing the Dirac operator for Robertson-Walker metrics to a constant warping function\cite{connes_chamseddine_spectral_action}. 
\end{remark}

\section{\sectiontitlefont{The structure maps for the Laplacian generated flow}}\label{section_structure_matrix}

From \cite{gakkhar2022quantum}, the heat semigroup $e^{-t\laplace^\End}$ is a quantum dynamical semigroup on $\End(E)\equiv E\tsr E^*$ with $e^{-t\laplace^\End}(\one)=\one$ for all $t$. We will work with the semigroup living on $\C(M)\subset \End(E)$. To derive the qsde associated to the heat semigroup, we 
start by computing the flow the structure matrix for the associated Evans-Hudson flow following the standard prescription (see \cite{goswami_qsp}): first we compute the kernel for the generator $\Lc={\laplace^\End}$ on the $\Ac_\infty = \C^\infty(M)$ acting on $\End(E)$ defined by $K_\Lc:X \times X\to \Bc(E\tsr E^*)$ for $X:= \Ac_\infty\times \Ac_\infty$, where $K_\Lc$ for any given any $\Lc:X\to X$ is defined by $$
X\times X \ni ((f_1, f_2), (g_1, g_2))\to \Lc(f_1^*f_2^*g_2g_1) + f_1^*\Lc(f_2^*g_2)g_1 - \Lc(f_1^*f_2^*g_2)g_1 - f_1^*\Lc(f_2^*g_2g_1) \in \C(M)
$$

For a basis of eigensections $(h_i)$ of $\laplace $, write $\sum_i h_i\tsr h_i^* = 1\in  \End(E)$. As $\grad_X\phi = df(X)\phi + f\grad_X(\phi)$, for $f_m, g_m\in \C(M)$, the contribution of $\grad_i\tsr \grad_i$ term to the kernel $((f_1, f_2), (g_1, g_2))\to \Lc(f_1^*f_2^*g_2g_1) + f_1^*\Lc(f_2^*g_2)g_1 - \Lc(f_1^*f_2^*g_2)g_1 - f_1^*\Lc(f_2^*g_2g_1)$ vanishes because $f_m, g_n$ are acting by multiplication on the $\grad_X\phi$ term and they all commute,  while for the term $df(X)\phi$, $d(f_1^*f_2^*g_2g_1) + f_1^*d(f_2^*g_2)g_1 - f_1^*d(f_2^*g_2g_1) - d(f_1^*f_2^*g_2)g_1$ vanishes as well since \begin{align}\label{eq_kernel_symmetry_kills_it}
(f_1^*d(f_2^*g_1g_2) &+ (df_1^*)f_2^*g_1g_2) + f_1^*d(f_2^*g_1)g_2 - f_1^*d(f_2^*g_2g_1) - d(f_1^*f_2^*g_2)g_1 \\
&= (df_1^*)f_2^*g_1g_2 + f_1^*d(f_2^*g_1)g_2 - ((df_1^*)(f_2^*g_2)g_1 +  f_1^*d(f_2^*g_2)g_1) = 0\nonumber
\end{align}

By commutativity of $\C(M)$, $1\tsr \laplace$ also has no contribution, since by convention the $f\cdot 1_{\End} = \sum_i(f\cdot h_i)\tsr h_i^*$. The only contribution to the kernel comes from the piece $\laplace\tsr 1$. Suppose $\phi$ is an eigensection with eigenvalue $\lambda$. Then with $f\in \C^\infty(M)$, we have \begin{align*}
\laplace f\phi &= -\sum_k \grad_{e_k}\grad_{e_k}f\phi = -\sum_k\grad_{e_k}\left( df(e_k)\phi + f\grad_{e_k}\phi\right)\\
&= -\sum_k\left(d^2f(e_k, e_k)\phi + 2df(e_k)\grad_{e_k}\phi + \grad_{e_k}\grad_{e_k}\phi\right)= -\sum_k\left(d^2f(e_k, e_k)\Phi + 2df(e_k)\grad_{e_k}\Phi\right) + f \lambda \phi
\end{align*}
Again by commutativity of $\C(M),$ $\lambda f \phi$ term does not contribute to the kernel, while a the same calculation as equation~\ref{eq_kernel_symmetry_kills_it} establishes that $\sum_k 2df(e_k)\grad_{e_k}\phi$ contributes zero. The only contribution to the kernel comes from the term $-\sum_k d^2f(e_k, e_k)\Phi$. Computing the kernel for $\sum_k d^2f(e_k, e_k)$ gives $-2\left(\sum_{k}da_{1}(e_k) db_{1}(e_k) \right)a_{2} b_{2}$.


\subsection{The Kolmogorov decomposition}

Following Goswami-Sinha construction of the flow generator, we compute the Kolmogorov decomposition for the kernel for $\Lc=-\laplace^{\End{}}/2$ (to avoid factors of $-2$), the data we need will be derived from the structure theory of this kernel,  \begin{align}\label{eqn_full_rw_kernel}
K_{\Lc}((a_1, a_2), (b_1, b_2))= \left(\sum_{k} 
da_{1}(e_k) db_{1}(e_k) \right)a_2b_2 
\end{align}
As in \cite[Thm~2.2.7]{goswami_qsp}, the Kolmogorov decomposition is taken to be the reproducing kernel Hilbert space $$
R_{\Lc} = \closure{\LinSpan\{K_{\Lc}(\cdot, b)u:b\in C^\infty(M), u\in E\}}
$$
with map $V:\C^\infty(M) \to \Bc(\E, R_\Lc)$ for $x\in \C^\infty(M)$ given by $$
V(x):E \to R_\Lc, E\ni u \to K_\Lc(\cdot, x)u \in R_\Lc
$$ 
By definition, $K_\Lc(\cdot, x)u$ is total in $R_{\Lc}$ making the decomposition minimal. 
\begin{remark}\label{rem_reproducting_hilbert_space_structure} Comparing with equation~\ref{eqn_full_rw_kernel} the term $a_2b_2\sum_k da_1(e_k)db_1(e_k)$ in local coordinates can be interpret as the form $a_2\sum_k da_1(e_k)de_k$ evaluated on the vector field $b_2\sum_k db_1(e_k)\frac{\partial}{\partial e_k}$. So each $K_{L}(\cdot, b)u$ is valued in $\Sect(TM)\tsr L^2(M, E)$ and for each $x=(a_1, a_2)\in \C^\infty(M)\times \C^\infty(M)$, $V(x)$ is a $1$-form acts by contracting with $\Sect(TM)$ component of $K_{\Lc}(\cdot, b)u$. 
\end{remark}
\subsection{The structure matrix}

The structure matrix\FRem{\cite[pg~138]{goswami_qsp}\cite[pg~351]{belton2015algebraic}} is the map $\Theta:\Ac_0\to  \Bc(H\tsr(\CC\oplus k_0))$ for a dense subalgebra $\Ac_0\subset \Ac$, and Hilbert spaces $H, k_0$ (alternatively, a map $\Theta:\Ac_0\to \Ac_0\tsr\Bc(\CC\oplus k_0), \Ac_0\subset \Bc(H)$) given by \begin{align}\label{eq_flow_generator_definition}
\Theta(x) = \begin{pmatrix}
\Lc(x)  & \delta^\dagger(x)\\
\delta(x) & \sigma(x)
\end{pmatrix}
\end{align}
where $\Lc:\Ac_0\to \Ac_0$ is a $*$-linear map  such that $\Lc(xy)-\Lc(x)y-x\Lc(y) =\delta^\dagger(x)\delta(y)$, $\delta:\Ac_0\to \Ac_0\tsr k_0$ is a $\pi$-derivation, $\Theta(1)=0$. Additionally, $\sigma(x)=\pi(x)-x\tsr 1$. The maps $\laplace, \delta, \sigma$ satisfy these conditions iff the first order Ito product formula holds\cite[lemma~2.2]{belton2015algebraic}. 

For the structure matrix $\Lc$ is the semigroup generator $-\laplace^\End$. Following \cite{goswami_qsp}, $\delta, \sigma$ are extracted from the minimal Kolmogorov decomposition: the decomposition $(R_{\Lc}, V)$, induces the following maps on $\Ac_\infty:=\C^\infty(M) \subset \C(M)=\Ac$, \begin{align*}
\rho:\Ac_\infty\to \Bc(R_\Lc), \rho(x)(V(\cdot, b)u) = V(\cdot, xb)u \\
\alpha:\Ac_\infty\to \Bc(E, R_{\Lc}), \alpha(x) = V(x, 1)
\end{align*}
With remark~\ref{rem_reproducting_hilbert_space_structure} and  equation~\ref{eqn_full_rw_kernel} in mind, $\rho(f)$, $f\in \C^\infty(M)$, is multiplication by $f$ on $R_{\Lc}$ while $\alpha(f)$ acts by contraction with $1$-form $\sum_kdf(e_k)de_k$. The representation $\rho$ is the identity map: $C^\infty(M)$ is interpreted as acting by multiplication on $R_{\Lc}$, and $\alpha$ is a derivation. 

The construction of the structure maps proceeds as in \cite[Thm~6.6.1]{goswami_qsp}. To start define the Hilbert $\Ac$-module $E = \closure{\{\alpha(x)y:x,y\in \Ac_\infty\}}$ where the closure is with respect to operator norm for $\Bc(\Hs_E, R_{\Lc})$. $\Ac$ has right action on $E$ by multiplication (where the norm density of $\Ac_\infty\subset\Ac$ is utilized) and the $\Ac$-valued inner product is $E\times E \ni (a, b)\to \ip{a, b} = a^*b$. 

Now $E$ can be identified $\alpha(f)g\in E$  can be identified with $g\tsr df \in \Ac\tsr_{\CStar} W\subset  \Ac\tsr_{\CStar} L^2(\Omega^1(M))$ where $W = \closure{\{df:f\in \C^\infty(M)\}}$ (note $M$ is compact) is a closed subspace of $L^2(\Omega^1(M)) = W\oplus W^{\perp}$ (the $\Ac$ valued innerproduct is pointwise tensor contraction). However, $\Ac\tsr_{\CStar} L^2(\Omega^1(M)) = L^2(\Omega^1(M))$ since $\Ac=\C(M)$ the tensor product is $\C(M)$-balanced: $a\tsr\omega \in \Ac\tsr_\CStar L^2(\Omega^1(M))$ is simply $1\tsr a\omega$. 

This choice of the Hilbert space $L^2(\Omega^1(M))$ simplifies the application of Kasparov stabilization theorem\cite[Thm~4.1.10]{goswami_qsp}. Now the Hilbert-$\CStar$-module $\Ac\tsr_{\CStar}L^2(\Omega^1(M)) = \Ac\tsr L^2(\Omega^1(M))=L^2(\Omega^1(M))$, and by Kasparov stabilization theorem, there's a unitary map $t':E\oplus \Ac\tsr L^2(\Omega^1(M)) \to \Ac\tsr L^2(\Omega^1(M))$. In this case it can be explicitly computed, but it turns out not to matter.

The $t'$ yields a unitary embedding $t:E\to \Ac\tsr_{\CStar}L^2(\Omega^1(M)), t=t'|_E$. Note that since the $\CStar$-module $\Ac\tsr_{\CStar}L^2(\Omega^1(M))$ has a basis, one does not need Kasparov's stabilization theorem and an abstract embedding, everything can be done explicitly. 

Define $\delta(x)=t(\alpha(x))$. Now $\rho$ induces a left action $\hat\rho$ on $E$,  $\hat \rho(x)(\alpha(y)) = (\alpha(xy)-\alpha(x)y)$. But as $\alpha$ is a $\rho$-derivation, $\hat\rho(x) = x\alpha(y)$, so $\hat \rho(x)$ is multiplication by $x$ and is again identity representation of $\Ac$ acting by multiplication on sections on the bundle. Set $\pi(x) = t\hat\rho(x)t^*$, again $\pi=\one$ (so the explicit form of $t$ does not come into play). Therefore, with equation~\ref{eq_flow_generator_definition} in reference,\begin{itemize} 
\item The multiplicity space $k_0 = L^2(\Omega^1(M))$
\item $\delta: \Ac \to \Ac\tsr k_0, \delta(f) = 1\tsr df$. 
  
\item The requirement that $\delta(x)^*\delta(y) = \Lc(x^*y) - x^*\Lc(y) - \Lc(x)^*y$ forces $\delta^\dagger(f) := \delta(f^*)^*$ to act by $\delta(f)^\dagger(1\tsr \omega)= \inner{df, \omega} = \sum_idf(e_i)\omega(e_i)\in \C(M)$ (using $f=f^*$); so 
$\delta(f)^\dagger(a\tsr \omega)= a\inner{df, \omega}$ works. 
\item $\sigma=0$ since $\sigma(x)=\pi(x)-x\tsr \one_{k_0}$ since $\pi:\Ac_0\to \Ac_0\tsr \Bc(k_0)$ is identity
\item $\Lc$ is the generator $-\laplace^\End/2$ for the semigroup $e^{-t\laplace^\End/2}$.
\end{itemize} Therefore, the structure matrix for the Laplacian generated flow is the map
\begin{align}\label{eq_structure_matrix}
\Ac_0 \ni f\to \Theta(f) = \begin{pmatrix}
-\laplace^\End(f\cdot \one)/2  & \delta^\dagger(f)\\
1\tsr df & 0
\end{pmatrix} \in \Bc(H\tsr(\CC\oplus k_0)) 
\end{align}

\section{Growth bounds and continuity}\label{section_growth_bounds}

Let $(M, g)$ be a Ricci-flat Riemannian manifold with Levi-Civita connection $\grad$. Since action of $\laplace^\End$ on $\C(M)$ can be identified with the Laplace-Beltrami operator $\laplace^M\equiv \laplace$ (and the lift, the rough Laplacian, to the tensor bundle), we will work with $\laplace$. Let $\phi_i$ be the an eigenfunction of the Laplace-Beltrami operator, $\laplace = -g^{ab}\grad_a\grad_b$, with eigenvalue $\lambda_i^2$, $\laplace \phi_i = \lambda_i^2 \phi_i$. Define $\Fc$ as the generated by finite products and sums of $\phi_i$'s, $$
\Fc = \LinSpan\left( \prod_{i\in [k]}\phi_i: k\in \NN\right)
$$
$\Fc$ is norm-dense inside $\C(M)$, and $\Fc_1:=\{\prod_{i\in [k]}\phi_i:k=1\}\subset \Fc$ forms a basis for $\C(M)$. Set $d\Fc =\{df:f\in \Fc_1\}\subset k_0$ as a norm-dense subspace, $(k_0)_\infty$. The constructions from \namecite{goswami_qsp}, \namecite{belton2015algebraic} both proceed by controlling the growth of the flow generator on a dense algebra, a role here played by $\Fc$. Note we are only interested in bound on the $\laplace$, we can work with either $\laplace=\pm g^{ab}\grad_a\grad_b$ to avoid tracking signs if needed. 

\subsection{Growth of Laplacian iterates}

For $x\in M$ fix a small neighborhood of $x$, $V$, with $U\subset V$ open, such that $\bar U\subset V$, with coordinates $(x^i)$ and $(\partial_{x^i})\equiv \partial_i$ the coordinate vector fields. For any multi-index $\alpha$ we denote $\partial^\alpha = \partial_{\alpha_k}\dots \partial_{\alpha_1}$, and same for $\grad_\alpha$ with $\grad_i\equiv\grad_{\partial_i}$. Ricci-flatness implies for all $f\in C^\infty(M)$ (see, for instance, \cite[lemma~1.36]{chow2023hamilton}), \begin{align}\label{eq_commutator_is_ricci_tensor}
[\laplace, \grad_i](f) = \txt{Ric}_i^k\grad_k(f) = 0
\end{align}
On $(p,q)$-tensors $s, s'$ there's a natural innerproduct by contraction with $g^{ij}, g_{ij}, \inner{s, s'} = g^{i_1j_1}\dots g_{m_1n_1}{s'}^{m_1\dots}_{i_1\dots}s^{n_1\dots}_{j_1\dots}$. The Levi-Civita connection has a lift to the tensor bundle and an associated connection Laplacian, both also denoted $\grad, \laplace$. Denote\FRem{\cite[\S~2.1]{bin2004derivatives}} by $\grad^k u$ the $k^{th}$-covariant derivative and define the point-wise length with the innerproduct\cite[\S~2.2.1]{hebey_sobolev_ams}:\begin{align}\label{label_length_of_gradient_xu2004}\ell(\grad^ku)^2 =  g^{i_1j_1}\dots g^{i_kj_k} (\grad^ku)_{i_1\dots i_k}(\grad^ku)_{j_1\dots j_k} = \inner{\grad^ku, \grad^ku} \end{align}

\begin{lemma}\label{lemma_commutator_laplace_gradient_i}For all $f\in \C^\infty(M)$,  $$[\laplace, \grad_i]\grad_{j_k}\grad_{j_{k-1}}\dots\grad_{j_1}(f) = 0$$
\end{lemma}
\begin{proof} Using\FRem{(it's just $k$-tensor being evaluated on a basis element $\partial_{i_1}\tsr \partial_{i_2}\dots \tsr \partial_{i_k} $)} \begin{align}
(\grad^ku)_{i_1\dots i_k} = \grad_{i_1}\grad_{i_2}\dots \grad_{i_k}(u)
\end{align} and that $g^{ij}$ commutes with $\laplace, \grad$ since $\grad$ is metric compatible, \begin{align}
\ell([\laplace, \grad_i]\grad_{j_k}\grad_{j_{k-1}}\dots\grad_{j_1}(f))^2 &= g^{ii'}g^{j_1j_1'}\dots g^{j_kj_k'}([\laplace, \grad_i]\grad_{j_k}\dots\grad_{j_1}(f))([\laplace, \grad_i']\grad_{j'_k}\dots\grad_{j'_1}(f))\nonumber \\
&= g^{ii'}([\laplace, \grad_{i}][\laplace, \grad_{i'}])(g^{j_1j_1'}\dots g^{j_kj_k'}\grad_{j_k}\dots\grad_{j_1}(f)\grad_{j'_k}\dots\grad_{j'_1}(f))
\end{align}
Since $h:=g^{j_1j_1'}\dots g^{j_kj_k'}\grad_{j_k}\dots\grad_{j_1}(f)\grad_{j'_k}\dots\grad_{j'_1}(f)\in \C^\infty$, from equation~\ref{eq_commutator_is_ricci_tensor}, $$
\ell([\laplace, \grad_i]\grad_{j_k}\grad_{j_{k-1}}\dots\grad_{j_1}(f))^2 = g^{ii'}([\laplace, \grad_i]([\laplace, \grad_{i'}])(h) = 0. 
$$
\end{proof}

\begin{remark} Similar bounds can be obtained for Einstein manifolds where the Ricci tensor is proportional to the metric. But since the analysis here needs a parallel spinor, the Evans-Hudson qsde over Einstein manifolds are not considered. 
\end{remark}

\begin{prop}\label{prop_laplace_bound_by_grad2} For all $f\in \C^\infty(M)$, $\norm{\laplace(f)}\leq \sqrt{\dim M}\ell(\grad^2(f))$ and  $ \norm{\laplace^k(f)}\leq (\dim M)^{k/2}\ell(\grad^k(f))$\end{prop}
\begin{proof} Using  $\ell(g_{ab})^2 = g^{ac}g^{bd}g_{ab}g_{cd} = g^{cd}g_{cd}=\dim M$ $$
\laplace f = g^{ab}\grad_a\grad_b f = g^{ca}g^{bd}g_{cd}\grad_a\grad_b f = \inner{g_{cd}, \grad_a\grad_b f}
\leq \ell(g_{cd})\ell(\grad_a\grad_b f) \leq \sqrt{\dim M} \ell(\grad_a\grad_b f)$$
where we used Cauchy-Schwarz. The bound $(\dim M)^{k/2}\ell(\grad^k(f))$ follows identically using $\ell(g_{i_1j_1}\dots g_{i_kj_k})^2 =(\dim M)^k$. 
\end{proof}

The growth of Laplacian and its powers is clear on $\phi_i$'s, since $\laplace^k(\phi_i) = \lambda_i^{2k}\phi_i$. Controlling the growth on products of $\phi_i$'s will require control over \begin{align}
\inner{\grad^k \phi_i,\grad \phi^k_j}(x)\leq \sqrt{\norm{\grad^k \phi_i(x)}_2\norm{\grad^k \phi_j(x)}} = \ell(\grad^k \phi_i(x))\ell(\grad^k \phi_j(x))
\end{align}
(using  eq~\ref{label_length_of_gradient_xu2004} with $k=1$ and Cauchy-Schwartz inequality). To see this note how Laplacian and its iterated powers act on products of $\phi_i$'s. 

\begin{prop}\label{lemma_laplacian_product_rule}[Laplacian on products] For $\phi_j,\phi_i$,
\begin{enumerate}[label=\roman*]
\item $\laplace(\phi_i\phi_j) = \phi_i\laplace(\phi_j) + \phi_i\laplace(\phi_j) + 2\inner{\grad \phi_i, \grad \phi_j}$
\item For any $k$, \begin{align}\label{eq_laplace_of_grad_innerproduct}\laplace\inner{\grad^k\phi_i, \grad^k\phi_j} = \inner{\grad^k\laplace\phi_i, \grad^k\phi_j} + \inner{\grad^k\phi_i, \grad^k\laplace\phi_j} + 2\inner{\grad^{k+1}\phi_i, \grad^{k+1}\phi_j}\end{align}
\end{enumerate}
\end{prop}
\begin{proof} These are straightforward computations, see appendix \ref{section_laplacian_products} for details.
\end{proof}

\noindent Therefore, it's sufficient to bound $\norm{\grad^k(\phi_j) }_{\infty}$. For $\phi_j$ with $\lambda_j \geq 1$, the bound follows from \cite[lemma~2.7]{bin2004derivatives}: 
\begin{theorem}[\namecite{bin2004derivatives}] With $\lambda\geq 1$, \begin{align}
\sum_{\lambda_j\in [\lambda, \lambda+1]} \ell(\grad^k\phi_j(x))^2 \leq C_{M, x}\lambda^{2k+\dim M-1}\end{align}
\end{theorem}
The following corollary is immediate since $M$ is compact and $C_M =(\lambda_j+1)^{\dim M -1}\sup_{x\in M}C_{M, x}$ and $K_j=(\lambda_j+1)^2$ can be used.  
\begin{corollary}\label{corollary_polynomial_ggrowth_kgradients} For $M$ compact, for any $\phi_j$ with $\lambda_j\geq 1$,  $
\ell(\grad^k \phi_j(x))^2 \leq C_{M}K_j^k 
$.\end{corollary}

One expects that $\norm{\grad^k \phi_j}_{L^2(M)}$ should decay to zero for $j$ with $\lambda^2_j<1$, and since $\phi_j$'s are smooth this is enough to establish a uniform bound. However, this will need to be leveraged locally and the boundary for the local chart will need to be taken into account. Recall the integration on parts formula for tensor fields\cite[pg~50,149]{lee_riemannian}  when $M$ does have a boundary, \begin{align}\label{eq_integration_by_parts_tensor_fields} 
\int_M \inner{\grad F, G} dV_g = \int_{\partial M}\inner{F\tsr N^{\flat}, G}dV_{\hat g} - \int_M \inner{F , \Div(G)}dV_g
\end{align}
where $\hat g$ is the induced metric on $\partial M, dV_g, dV_{\hat g}$ the associated volume forms, $\cdot^{\flat}$ the musical isomorphism, $N$ the outward unit normal at $\partial M$, and $F, G$ tensor fields, $\Div(G) = \Tr_g(\grad G)$, the trace being over the last two indices. Note if $G=\grad H$ then, $-\Div(G)=\laplace(H)$.

\begin{prop}\label{prop_gradients_are_contractive} For $u=\phi_j$, with $j$ such that $\lambda^2_j < 1$, $\norm{\grad^k u }_\infty  \leq \norm{\grad^{k-1}u }_\infty$
\end{prop}
\begin{proof} Suppose for some $x\in M$, $\inner{\grad^{k}u(x), \grad^{k}u(x)} - \inner{\grad^{k-1}h(x), \grad^{k-1}u(x)} > 0$. Then since $u$ is smooth, there exists an open neighborhood $U$ of  $x$ such that on $U$, $\inner{\grad^{k}u, \grad^{k}u} - \inner{\grad^{k-1}u, \grad^{k-1}u} > 0$. Now let $\psi$ be such that $\txt{supp}(\psi)\subset U$ is compact, $\psi\geq 0$ on $U$ and $\psi > 0$ on open $V\subset U$, then
\begin{align}\label{eq_suppose_grad_u_does_not_decay}
\int_M \psi\inner{\grad^{k}u, \grad^{k}u} - \psi\inner{\grad^{k-1}u, \grad^{k-1}u} dV_g &= \int_U \inner{\psi\grad^{k}u, \grad^{k}u} - \psi\inner{\grad^{k-1}u, \grad^{k-1}u} dV_g> 0 \end{align}
Now $\inner{\psi\grad^{k}u, \grad^{k}u} =  \inner{\grad(\psi\grad^{k-1}u), \grad^{k}u} - \inner{\grad\psi\cdot \grad^{k-1}u, \grad^{k}u}$, and for the first term $$
\int_U 
\inner{\grad(\psi\grad^{k-1}u), \grad^{k}u}
dV_g = \int_{\partial U} \inner{\cdot, \cdot}dV_{\hat g} + \int_U\inner{\psi \grad^{k-1}u, -\Div(\grad^{k}u)}dV_g
$$
where $\int_{\partial U} \inner{\cdot, \cdot}dV_{\hat g} = 0$ since $\psi=0$ on $\partial U$, while $-\Div(\grad^{k}u) = \laplace\grad^{k-1}u = \grad^{k-1}\laplace u$ using $M$ is Ricci flat. Therefore, we have\begin{align*}
\int_M \psi\inner{\grad^{k}u, \grad^{k}u}  
&= \int_U \psi\inner{\grad^{k-1}u, \grad^{k-1}\laplace u} - \int_U \inner{\grad\psi\cdot  \grad^{k-1}u, \grad^{k}u} dV_g\\
&= \lambda^2\int_U \psi\inner{\grad^{k-1}u, \grad^{k-1} u} - \int_{\txt{supp}(\grad\psi)}  \inner{\grad\psi\cdot \grad^{k-1}u, \grad^k u} dV_g
\end{align*}
This yields 
\begin{align}0 < \int_U \psi&\inner{\grad^{k}u, \grad^{k}u} - \psi\inner{\grad^{k-1}u, \grad^{k-1}u} dV_g\label{eq_bound_over_U_with_psi} \\
&=  \lambda^2\int_U \psi\inner{\grad^{k-1}u, \grad^{k-1} u} - \int_{\txt{supp}(\grad\psi)}  \inner{\grad\psi\cdot \grad^{k-1}u, \grad^k u} dV_g - \int_U\psi\inner{\grad^{k-1}u, \grad^{k-1}u} dV_g \nonumber \\
&= \int_U\psi(\lambda^2 - 1) \norm{\grad^{k-1}u}^2 dV_g -\int_{\txt{supp}(\grad\psi)}   \inner{\grad\psi\cdot \grad^{k-1}u, \grad^k u} dV_g \label{eq_why_psi_must_be_positive}
\end{align}
Define the linear functional $\omega(\psi) := \int_{\txt{supp}(\grad\psi)} \inner{\grad\psi\cdot \grad^{k-1}u, \grad^k u} dV_g$. Note \begin{align}
\inner{\grad\psi\cdot \grad^{k-1}u, \grad^k u} 
&=   g^{i_1j_1}\grad_{i_1}\psi\left(g^{i_2j_2}\dots g^{i_kj_k}(\grad^{k-1}u)_{i_2\dots i_k} (\grad^{k}u)_{j_1j_2\dots j_k}\right)\nonumber \\
&= g^{i_1j_1}\grad_{i_1}\psi\left(g^{i_2j_2}\dots g^{i_kj_k}(\grad_{i_2}\dots\grad_{i_k}u) (\grad_{j_1}\grad_{j_2}\dots\grad_{j_k}u)\right)\nonumber\\
& = \textstyle{\frac{1}{2}}g^{i_1j_1}\grad_{i_1}\psi\grad_{j_1}(\ell(\grad^{k-1}u, \grad^{k-1}u)^2) = \textstyle{\frac{1}{2}}\inner{\grad\psi, \grad\ell(\grad^{k-1}u, \grad^{k-1}u)^2}\label{eq_mixed_grad_term_as_ell}
\end{align}
By showing that there exists a $\psi$ that makes $\omega(\psi) \geq 0$, since $(\lambda^2-1) < 0$, from equation~\ref{eq_why_psi_must_be_positive} it will follow that  equation~\ref{eq_bound_over_U_with_psi} cannot hold.

Assume that $U$ is small enough to be covered by a geodesic normal coordinates, and consider polar coordinates on $U$ centered at $x$. Define $\tau^c_s$ on $U$ for $c,s\in \RR_{>0}$ such that $\tau_s(x)=c$ and then decays linearly in radially outwards direction with slope $-s$ to $0$ at $\partial B_{R}(x)$ with $c,s$ such that $\closure{\txt{supp}(\tau^c_s)}\subset U$, $R$ depending on $c,s$. Then $\tau^c_s$ is continuous, piecewise continuously differentiable, with compact support in $U$, so weakly-differentiable, and $\grad \tau^c_s = -s\one_{\B_{R}(x)}$ (there's enough slack to work with mollified versions of $\tau$'s, but weak-differentiablility suffices for simplicity).  If for some $\tau^c_s$, $\omega(\tau^c_s) \geq 0$ then that $\psi=\tau^c_s$ is the required $\psi$. 

If not, then $\omega(\tau^c_s) < 0$ for all $c$ small enough to have support in $U$. 
By rescaling wlog assume $c=s=1$, and set $\tau_1:=\tau^1_1$ ( otherwise the constants are messy). For such $\tau_1$, define $\tau'_1$ such that $\tau'_1(x)=0,$ and $\tau'_1$ increases linearly to $1$ at $\partial B_1(x)$, and outside of $\closure{B_1(x)}, \tau'_1 = 0$. Then $\omega(\tau'_1) = -\omega(\tau_1) = \delta > 0$ since $\grad\tau_1' = -\grad\tau_1$ on $\txt{supp}(\grad\tau_1') = \txt{supp}(\grad\tau_1)$. It remains to make $\tau_1'$ continuous without changing $\omega(\tau'_1)$ too much. For this set $\tau_{1,r}'' = \tau_1'$ on $B_1(x)$, $\tau_{1,r}'' =0$ on $B_{1+r}(x)^c$, and on $B_{1+r}(x)^c \setminus  B_{1}(x)$, $\tau''_{1,r}$ decays linearly to $0$ on $\partial B_{1+r}(x)$. Finally, since for all $r>0$ small enough, $\tau''_{1,r}$ is piecewise continuous, continuously differentiable and compactly supported in $U$, it remains to check for any $\eps>0$ there exists $r_\eps>0$ such that for all $r<r_\eps$, $\norm{\omega(\tau'_s) - \omega(\tau''_{s, r})}\leq \eps$. Note that \begin{align}
 2\norm{\omega(\tau'_1) - \omega(\tau''_{1, r})} =
 \norm{\int_{B_{1+r}(x)(x)\setminus B_{1}(x)}  \inner{\tau_s'', -\Div(\grad(\ell(\grad^{k-1}(u))^2))} dV_g}
\label{eq_finally_make_tau_difference_small}
\end{align}
using equation~\ref{eq_mixed_grad_term_as_ell} and so $\norm{\omega(\tau'_1) - \omega(\tau''_{1, r})}\leq \eps$ for small $r$
because  
$-\Div(\grad(\ell(\grad^{k-1}(u))^2))=\laplace(\ell(\grad^{k-1}(u))^2),$ $ \tau''_{s, r}$ being continuous (with $M$ compact) are bounded. 
\end{proof}

The idea above generalizes to all $\phi_i$'s giving a simple proof for growth bounds on covariant derivatives from \cite{bin2004derivatives} in the Ricci flat setting. 

\begin{corollary}
For $u=\phi_j$, with $\lambda^2:=\lambda^2_j$, $\norm{\grad^k u }_\infty  \leq \lambda^2\norm{\grad^{k-1}u }_\infty$
\end{corollary}
\begin{proof} Assume not, then on some open $U\subset M$, for all $x\in U$, for some $c>1$, $$
\inner{\grad^{k}u(x), \grad^{k}u(x)} - c\lambda^2\inner{\grad^{k-1}h(x), \grad^{k-1}u(x)} > 0
$$ and as in proposition~\ref{prop_gradients_are_contractive} for some $\psi\geq 0$ compactly supported in $U, \psi >0$ on an open set, \begin{align}\label{eq_general_growthboundeq_1}
\int_U \inner{\psi\grad^{k}u(x), \grad^{k}u(x)}dV_g - \int_U c\lambda^2\psi\inner{\grad^{k-1}h(x), \grad^{k-1}u(x)}dV_g > 0 
\end{align}
 $$
 \txt{with }\int_U \inner{\psi\grad^{k}u(x), \grad^{k}u(x)}dV_g =  \int_U \inner{\grad(\psi\grad^{k-1}u(x)), \grad^{k}u(x)}dV_g - \int_U \inner{\grad\psi \cdot \grad^{k-1}u, \grad^{k}u}  dV_g
$$
Therefore, equation~\ref{eq_general_growthboundeq_1} becomes \begin{align*}
\int_U \lambda^2&\inner{\psi\grad^{k-1}u(x), \grad^{k-1}u(x)}dV_g - \int_U \inner{\grad\psi \cdot \grad^{k-1}u, \grad^{k}u}  dV_g
 - \int_U c\lambda^2\psi\inner{\grad^{k-1}h(x), \grad^{k-1}u(x)}dV_g \\
 &= \lambda^2(1- c)\int_U\inner{\psi\grad^{k-1}u(x), \grad^{k-1}u(x)}dV_g - \int_U \inner{\grad\psi \cdot \grad^{k-1}u, \grad^{k}u}  dV_g > 0
\end{align*}
But choosing $\psi$ as in~\ref{prop_gradients_are_contractive}, since $\lambda^2(1-c) < 0$, the last inequality cannot hold.
\end{proof}

\noindent Therefore, $\norm{\grad^k \phi_i }_\infty$ is uniformly bounded  for all $\phi_i$.
Collecting this with propositions~\ref{prop_laplace_bound_by_grad2} and \ref{lemma_laplacian_product_rule} along with corollary~\ref{corollary_polynomial_ggrowth_kgradients}, yields easily implies the following bounds.

\begin{theorem}\label{thm_laplacian_growthbounds}For all $k\in \NN, \phi\in \Fc$,\begin{enumerate}[label=\roman*]
\item There exists constants $C_\phi, K_\phi$ such that $\norm{\laplace^k(\phi)}\leq C_\phi K_\phi^k$ 
\item For any $k, N\in\NN$, the map $\laplace^k\tsr \one_{\mat_N}:\mat_N(\Fc)\to \mat_N(\C^\infty(M))$, $\norm{\laplace^k\tsr \one_{\mat_N}(\xi)} \leq C_{\xi}K^k_{\xi}$ for $\xi=[\xi_{ij}]\in \mat_N(\Fc)$. The same holds for $(\laplace\tsr \one_{\mat_N})^k$. 
\end{enumerate}
\end{theorem}
\begin{proof} Part $i)$ follows from estimates on $\laplace^k$, how $\laplace^k$ acts on products of $\phi_i$'s and $\grad$ being a derivation. For part $ii)$, by growth bounds on $\laplace^k$,  $\laplace^k(\xi_{ij})\leq C_{ij}M_{ij}^k$ for some constants $C_{ij}, M_{ij}$. so \begin{align} \norm{\laplace^k\tsr \one_{\mat_N}(\xi)} = \norm{[\laplace^k(\xi_{ij})]} \leq \max_i\left(\sum_j C_{ij} \cdot \left(\sum_j M_{ij}\right)^k\right). 
\end{align}
\end{proof}
\noindent Note that theorem~\ref{thm_laplacian_growthbounds}, $ii)$ is needed to guarantee that the map-valued process $J_t$ evaluated on exponential vectors is pointwise bounded. 

\section{Picard iterates for the Laplacian generated flow}\label{section_picard_iterates_convergence}

We now specialize the map-valued Evans-Hudson qsde’s to the Laplacian flow generator, $\Theta$: 
from equation~\ref{eq_structure_matrix}, \begin{align}
\Ac_\infty\supset \Fc \ni f\to \Theta(f) = \begin{pmatrix}
-\laplace(f)/2  & \delta^\dagger(f)\\
1\tsr df & 0
\end{pmatrix} \in \Bc(H\tsr(\CC\oplus k_0)) 
\end{align}
where $k_0= L^2(\Omega^1(M))$, the noise space, and $H=L^2(S)$, the square integrable spinors on which $\C(M)$ acts. Recall $(k_0)_\infty := d\Fc = \{df:f\in \Fc\}\subset L^2(\Omega^1(M)$. $k=L^2(\RR_+, k_0)$ For $f\in \CC\oplus k_0$. For $f\in k_0$, set $\hat f := 1\oplus f \in \CC\oplus k_0$ and identify $f$ with $0\oplus f\in \CC\oplus k_0$

Towards establishing the existence of the quantum flow of Evans-Hudson type under the much weaker regularity assumption and without the Frechet space scaffolding from \cite[Thm~8.1.38]{goswami_qsp}; the usual ideas from the literature can be leveraged, however, crucial estimates need to be made from ground up which we establish next. Only a sketch of known results into which these estimates are plugged are included how they made fit into the scheme. A crucial point is an extension of the squareroot trick, which is needed since in general the algebra $\Fc$, the algebra consisting of finite products and sums of eigenfunctions of $\laplace$, is not closed under squareroots since derivatives of squareroots grow factorially while on $\Fc$ the covariant derivatives satisfy growth bounds of type $\ell(\grad^k f)\leq C_fM_f^k$. For background on map-valued qsdes, we refer to the included appendix~\ref{section_mapvalued_qsdes}.  To start, recall the following estimates for map-valued processes, $a_\delta, a_\delta^\dagger, I_\Lc$ being the fundamental processes (appendix equations~\ref{eq_fundamental_processes1},\ref{eq_fundamental_processes2}): 
\begin{estimate}{\cite[Thm~5.4.7,8.1.37]{goswami_qsp}}\label{estimate_mapvalued_integrals}
For a map-valued integrable process $Y_s$, 
\begin{align}
\norm{\int_0^t Y_s\circ (a_\delta + I_\Lc)(ds)(x\tsr \ExpHVec(f))u}^2 &\leq e^t\int_0^t\norm{Y_s(\Lc(x)+ \inner{\delta(x^*), f(s)})\tsr \ExpHVec(f))u}^2ds\label{eq_annihilation_integral} \\
\norm{\int_0^t Y_s\circ (a^\dagger_\delta)(ds)(x\tsr \ExpHVec(f))u}^2 &\leq e^t\int_0^t\left(\norm{\wtilde{Y_s}(\delta(x)\ExpHVec(f))u}^2 + \norm{Y_s(\inner{f(s), \delta(x)})\ExpHVec(f)u}^2 \right)ds\label{eq_creation_integral}
\end{align}
\end{estimate}
There's the following characterization for an integrable map-valued process generated by the structure maps $\Lc=-\laplace/2, \delta$ from $\Theta$ through the Picard iteration scheme, the convergence of which will yields the solution to the qsde needed.
\begin{lemma}{\cite[lemma~8.1.37]{goswami_qsp}}\label{lemma_picard__recursive_estimate} Let $\Vc=\{(d\Fc \txt{-valued simple functions}\}$, $\ExpHVec(\Vc)$ the exponential vectors, and $J^{(0)}:\Fc\tsr \ExpHVec(\Vc)\to \Ac\tsr \Fock(k_0)$ be the identity map, then with $J^{(0)} =  \one$, \begin{align} , J^{(n+1)}(t) = \int_0^t J^{(n)}(s)\circ (a^\dagger_\delta + a_\delta
 + I_\Lc)(ds), \ \ J^{(n+1)}:\Fc\tsr \ExpHVec(\Vc)\to \Ac\tsr \Fock(k_0)
\end{align}
each $J^{n}$ is an a map-valued integrable process (by definition linear, but not necessarily completely smooth), Additionally, 
the following estimates hold 
\FRem{\cite[Thm~5.4.7,pg~139]{goswami_qsp}}, 
\begin{align}\label{eq_picard_iterates}
\norm{J^{(n+1)}_t(x\tsr \ExpHVec(f))u}^2 &\leq 2\left(\norm{\int_0^t J^{(n)}_s\circ (a_\delta + I_\Lc)(ds)(x\tsr \ExpHVec(f))u}^2  + \norm{\int_0^t J^{(n)}_s\circ (a^\dagger_\delta)(ds)(x\tsr \ExpHVec(f))u}^2  \right)
\end{align}
\end{lemma}

\begin{proof} The continuity requirements for existence of the integral are satisfied since for each fixed $\ExpHVec(f)$ and $a$, the maps $\delta, \laplace$ are bounded.
The inequalities follow from standard theory. The Laplacian flow generator satisfies much weaker assumptions, therefore the resultant process is not completely smooth. \end{proof}

The Picard iterates defined by $S_N(t) = \sum_{n\leq N} J_t^{(n)}(x\tsr \ExpHVec(f))$ can be shown to converge on the exponential vectors following same scheme as \cite[Thm~8.1.38]{goswami_qsp} after plugging in the following estimates which need to be obtained differently as $\Theta$ has much less regularity. To motivate the estimates we sketch the convergence arguments. 

The first term in \RHS for equation~\ref{eq_picard_iterates}, using the definition of map-valued integrals (see appendix~\ref{section_mapvalued_qsdes}.~\ref{eqs_deinition_mapvalued_integrals}) can be recursively expanded using via estimate~\ref{estimate_mapvalued_integrals}, inequality ~\ref{eq_annihilation_integral}: $$
e^t\int_0^t\norm{J^{(n)}_s(\Lc(x)+ \inner{\delta(x^*), f(s)})\tsr \ExpHVec(f))u}^2ds = e^t\int_0^t \norm{\int_0^s J^{(n-1)}_{s^1}\circ (a_\delta + I_\Lc)(ds^1)(x\tsr \ExpHVec(f))u}^2 ds
$$
Since $f$ is simple, all terms depending on $f$ in above can be uniformly bound by a constant $B_f$, and recursively applying inequality~\ref{eq_annihilation_integral} to the \RHS, till reaching $J^{(0)}=\one$ yields:  
\begin{align}\label{eq_annihilation_term_control}
\norm{J^{(n+1)}(t)(a\tsr \ExpHVec{f})}^2 &\leq (2e^tB)^n\norm{\ExpHVec(f)}^2 \int_0^t\int_0^s\int_0^{s_1}\dots \norm{\Phi_{\hat f(s)}(\Phi_{\hat f(s_1)}\dots \Phi_{\hat f(s_{n-1})}(x))} ds ds_1\dots ds_{n-1}\nonumber \\
\end{align}
where $\Phi_{f(s)}(x):=\Lc(x)+ \inner{\delta(x^*), f(s)}$. Since $f$ is simple, $\Range[f]$ is finite: for each $s\in [0, t]$, $f(s)\in 
\{d\zeta_i \in d\Fc, i\in [r], r\in \NN\}\equiv \Range[f]$. This means one must control \begin{align}\label{eq_needed_bound1}
\norm{\Phi_{\xi_{i_n}}(\Phi_{\xi_{i_{n-1}}}(\dots \Phi_{\xi_{i_1}}(x)))}
\end{align}
where each $\xi_{i_k} \in \Range[f]$. Similarly, for the second term in \RHS for equation~\ref{eq_picard_iterates}, from equation~\ref{eq_creation_integral} (also, see appendix equation~\ref{eq_ytilde_integrator}), we have $$
\norm{\int_0^t J^{(n)}_s\circ (a^\dagger_\delta)(ds)(x\tsr \ExpHVec(f))u}^2  \leq e^t\int_0^t\left(\norm{\wtilde{J^{(n)}}(\delta(x)\tsr\ExpHVec(f))u}^2 + \norm{J^{(n)}(\inner{f(s), \delta(x)})\ExpHVec(f)u}^2 \right)ds
$$
The $\|J^{(n)}(\inner{f(s), \delta(x)})\ExpHVec(f)u\|^2$ term is controlled exactly as equation~\ref{eq_annihilation_term_control} via a corresponding estimate on nested $\Phi'_{f(s)}(x):=\inner{f(s), \delta(x)}$. The $\|\wtilde{J^{(n)}}(\delta(x)\tsr\ExpHVec(f))u\|^2$ term lives on the Fock spaces it can be controlled by same recursive expansion using $\Phi''_{f(s)}(x):=\delta(x)\tsr\ExpHVec(f(s))$ since by definition (see appendix equation~\ref{eq_ytilde_integrator_in_terms_of_creation} and remark~\ref{remark_extending_ytilde}), integral with respect to $a^\dagger_\delta$ is given by -- \begin{align*}
\left(\int_0^t J^{(n)}(s)\circ (a^\dagger_\delta)(ds)   \right)(x\tsr \ExpHVec(f))u &= \int_0^t\left( \wtilde{J^{n}(s)}(\delta(x)\tsr \ExpHVec(f_s))u\right)ds\\ &= \int_0^t \left(\left( \int_0^s J^{(n-1)}(s)\circ (a^\dagger_\delta)(ds_1)   \right)(\delta(x)\tsr \ExpHVec(f_s))\right)u ds \\
&=\int_0^t \left( \int_0^s \wtilde{J^{(n-1)}}((\delta\tsr \one) (\delta(x)\tsr \ExpHVec(f_s)))u \right) ds_1 ds
\end{align*}
so one needs to bound both of the following \begin{align}\label{eq_needed_bound3}
\norm{\Phi'_{\xi_{i_n}}(\Phi'_{\xi_{i_{n-1}}}(\dots \Phi'_{\xi_{i_1}}(x)))},\  \norm{\delta^n(x)}
\end{align}
\begin{estimate}\label{estimate_creation_annihilation_terms} There exist constants, $C_{f, x}, M_{f, x}$ such that \begin{align*}
    i)\ &\norm{\Phi_{\xi_{i_n}}(\Phi_{\xi_{i_{n-1}}}(\dots \Phi_{\xi_{i_1}}(x)))} \leq C_{f, x} M_{f, x}^n\ \ \ ii) \ 
    \norm{\Phi'_{\xi_{i_n}}(\Phi'_{\xi_{i_{n-1}}}(\dots \Phi'_{\xi_{i_1}}(x)))} \leq C_{f, x} M_{f, x}^n \ \ \ 
    iii) \  \norm{\delta^m(x)} \leq C_{f, x} M_{f, x}^n
\end{align*}

That is, three terms (equations~\ref{eq_needed_bound1},\ref{eq_needed_bound3}) satisfy bound of form $CK^n$ for $C, K$ depending on the simple function $f$ and independent from $n$ (note the constants for $\Phi, \Phi', \Phi''$ may be different, but taking the maximum, yields a single pair that works). \end{estimate}

\begin{proof} $i)$ Define $a_n:= \Phi_{\xi_{i_n}}(\Phi_{\xi_{i_{n-1}}}(\dots \Phi_{\xi_{i_1}}(x)))$. Since each $\xi_{i_k}\in d\Fc$, let $\xi_{i_k} = dz_k$. Now $\Phi_{z_k}(a') = \laplace' a' + \inner{da', dz_k} = \laplace' a' + \inner{\grad a', \grad z_k}$ where $\laplace' := -\laplace/2$. Therefore, \begin{align*}
\norm{a_n} = \Phi_{z_n}(a_{n-1}) &= \laplace' a_{n-1} + \inner{\grad a_{n-1}, \grad z_n} = g_{i_nj_n}\left(-\frac{1}{2}\grad^{i_n}\grad^{j_n}a_{n-1} + \grad^{i_n}a_{n-1}\grad^{j_n}z_n\right)\\ 
&= g_{i_nj_n}g_{i_{n-1}j_{n-1}}\left(-\frac{1}{2}\grad^{i_n}\grad^{j_n}\left(-\frac{1}{2}\grad^{i_{n-1}}\grad^{j_{n-1}}a_{n-2} + \grad^{i_{n-1}}a_{n-2}\grad^{j_{n-1}}z_{n-1}\right)\right)\\
&\ \ \ \ \ \ + g_{i_nj_n}g_{i_{n-1}j_{n-1}}\left(\grad^{i_n}\left(-\frac{1}{2}\grad^{i_{n-1}}\grad^{j_{n-1}}a_{n-2} + \grad^{i_{n-1}}a_{n-2}\grad^{j_{n-1}}z_{n-1}\right)\grad^{j_n}z_n\right)
\end{align*}
Therefore, each recursive expansion of $\Phi_{z_k}$ doubles the number of terms. After $n$ recursive expansion there are $2^n$ terms, consisting of $n$ copies of the metric, $g_{i_nj_n}g_{i_{n-1}j_{n-1}}\dots g_{1_1j_1}$, contracting against the covariant derivatives acting on $x, z_1\dots z_n$. Note that $z_i$'s appear in increasing order left to right by the recurence structure, and $x$ is at the left most, making appearance only at the last step. Let $P_{n+1}[I]$ be a partition of set $I:=\{i_k,j_k:k\in [n]\}$ into $n+1$ \emph{ordered} sets $P_{n+1} = \{A_i\subset I:i\in [n+1]\}$, with some $A_i$'s possibly empty. For $A\in P_{n+1}[I]$ with $\grad^{A}, A=[a_1, a_2\dots a_{|A|}]$ denote the operator $\grad^{a_1}\grad^{a_2}\dots\grad^{a_{|A|}}$,and by $\grad_{A}$ the corresponding operator with indices lowered. Then after $n$ recursive expansions, each term corresponds to some partition $P_{n+1}[I] =\{A_i:i\in [n+1]\}$. Writing the tensor contraction as an explicit innerproduct, $$
g_{i_nj_n}g_{i_{n-1}j_{n-1}}\dots g_{1_1j_1} \grad^{A_1}x \grad^{A_2}z_1 \grad^{A_3}z_2\dots \grad^{A_{n+1}}z_n
= \inner{g_{i_nj_n}g_{i_{n-1}j_{n-1}}, \grad_{A_1}x \grad_{A_2}z_1 \grad_{A_3}z_2\dots \grad_{A_{n+1}}z_n} $$
where with empty $A_i$ meaning the term is missing. Applying Cauchy-Schwarz, and dropping the $-1/2$ factors as we are only interested in an upper bound on each term,
\begin{align*}
\inner{g_{i_nj_n}g_{i_{n-1}j_{n-1}}, \grad_{A_1}x \grad_{A_2}z_1 \grad_{A_3}z_2\dots \grad_{A_{n+1}}z_n} &\leq (\dim M)^{n/2}\sqrt{\inner{\grad_{A_1}x \grad_{A_2}z_1 \dots \grad_{A_{n+1}}z_n, \grad_{A_1}x \grad_{A_2}z_1  \grad_{A_{n+1}}z_n}}\\
&\leq (\dim M)^{n/2}\sqrt{\inner{\grad_{A_1}x, \grad_{A_1}x}\dots \inner{\grad_{A_{n+1}}z_n, \grad_{A_{n+1}}z_n}} \\
&\leq (\dim M)^{n/2}\ell(\grad^{|A_1|}x)\ell(\grad^{|A_2|})z_1\dots\ell(\grad^{|A_{n+1}|}z_n)
\end{align*}

For each $\phi\in\{x, z_n\}\equiv T_{x,f}$, there are constants $C_\phi, M_\phi$ such that $\ell(\grad^k\phi)\leq C_\phi M_\phi^k$. Wlog, assume each $C_\phi=1$, by absorbing it into $M_\phi$. Now since each $dz_n\in \Range[f], |\Range[f]|=r$, the choice $M=\max_{\phi\in T_{x, f}} M_\phi$ is independent of $n$, and using that $\sum_{i\in [n+1]}|A_i| = 2n$, yields $$
\inner{g_{i_nj_n}g_{i_{n-1}j_{n-1}}, \grad_{A_1}x \grad_{A_2}z_1 \grad_{A_3}z_2\dots \grad_{A_{n+1}}z_n} \leq ((\dim M)^{1/2}M^2)^n 
$$
And finally accounting $2^n$ terms with above inequality holding for each, yields the needed bound $$
\norm{a_n} \leq (2(\dim M)^{1/2}M^2)^n
$$

\noindent $ii)$ This follows identically to part $i)$. 

\noindent $iii)$ Since $\delta:\Ac\to \Ac\tsr k_0, x\to 1\tsr dx$, therefore $\delta^2(x) = (\delta\tsr \one)(\delta(x)) = (\delta\tsr \one)(1\tsr dx) = 0$. So the bound is simply $$\sup_{s\in [0, t]}\norm{\ExpHVec(f_s)}\norm{dx}$$\end{proof}

With these estimates available, the following results are easily obtained exactly as in \cite{goswami_qsp}. 

\begin{theorem}\label{app_thm_picard_iterates_on_F} Suppose $a\in \Fc$, and $f\in \Uc$ where $\Uc$ is the set of simple function taking values in $(k_0)^\Fc_\infty$ such that $f_t\in \{\xi_i:i\in [r]\}$ for all $t$. Define \begin{align} J^{(0)} =  \one, J^{(n+1)} = \int_0^t J^{(n)}(s)\circ (a^\dagger_\delta + a_\delta
 + I_\Lc)(ds)
\end{align}
where $\delta, \Lc$ from the structure matrix $\Theta$ 
then\begin{enumerate}
\item $J(t) = \sum_n J_t^{(n)}(a\tsr \ExpHVec(f))$ converges
\item $J(t) = \one +\int_0^t J(s)(a^\dagger_\delta + a_\delta
 + I_\Lc)(ds)$ holds
 \item For all $a_1, a_2\in \Fc, f_1, f_2$ simple and $(k_0)^\Fc_\infty$-valued, $v_1, v_2\in \Hs$, \begin{align*}
 \inner{J_t(a_1\tsr \ExpHVec f_1)v_1, J_t(a_2\tsr \ExpHVec f_2)v_2} &= \inner{v_1\tsr\ExpHVec f_1,  J_t(a_1^*a_2\tsr \ExpHVec f_2)v_2} \\
 J(t)(\one \tsr \ExpHVec f_1 )v_1 &= v_1\ExpHVec f_1
 \end{align*}
\end{enumerate}
\end{theorem}  
\begin{proof} From the estimate~\ref{estimate_creation_annihilation_terms}, $S_N(t) = \sum_{n\leq N} J_t^{(n)}(a\tsr \ExpHVec(f))$ is convergent since for each $J_t^{(n)}(a\tsr \ExpHVec(f))$, $$
\norm{J_t^{(n)}(a\tsr \ExpHVec(f))}\leq C_{f, x}(2e^tBM_{f, x})^n/n!
$$ This yields part \emph{1.} Part \emph{2.} is by the definition of the stochastic integral, and \emph{3.} is standard theory as in \cite[Thm~8.1.38]{goswami_qsp}.
\end{proof}
It remains to check that the integral extends from $\ExpHVec(\Uc)$ to $\Fock(k_0)$. For this we need to extend the squareroot trick using regularity of the flow generator. 

\subsection{The extended squareroot trick}

It follows from quantum Ito formulae that for $u, v\in H$, $h, f\in \Vc\subset L^2(\RR_+, k_0)\equiv k$, \begin{align}
\inner{J_t(a\tsr \ExpHVec(f))u,  J_t(b\tsr \ExpHVec(h))v } &= \inner{u\ExpHVec(f),  J_t(a^*b\tsr \ExpHVec(h))v }\label{eq_factorization_jt} \\
J_t(1\tsr \ExpHVec(f))u &= u\ExpHVec(f)
\end{align}
Define $$
j^n_t(a)(v\ExpHVec f) := J^{(n)}_t(a\tsr \ExpHVec f)v
$$
so $j^n_t$ is unital with the factorization property (equation~\ref{eq_factorization_jt}), and $j^n_t(a)$ is a linear operator on a dense subspace $\Kc:= H\tsr \ExpHVec(\Vc)\subset H\tsr \Fock(k)$. For any $v, f$, $j_t^n$ is bounded pointwise on $\Fc$. 

\begin{prop}\label{prop_sobolev_norm_bound_on_theta} For all $a\in \Fc$ there exists $K$ such that $$\norm{\Theta(a)} \leq K\norm{a}_{W^{2, \infty}}$$
where $\norm{\cdot}_{W^{2, \infty}}$ is the Sobolev norm $\norm{a}_{W^{2, \infty}} = \norm{a}+ \norm{\grad a} + \norm{\grad^2(a)}$
\end{prop}

\begin{proof}
For $a, \psi\in \Fc, f'\in d\Fc$ with $f'=df, f\in \Fc$ and $w\in \CC$, \begin{align}\label{eq_sobolev_bound_theta}
\begin{pmatrix}
-\laplace(a)/2  & \inner{\cdot, da} \\
1\tsr da & 0
\end{pmatrix}\begin{pmatrix} \psi\tsr w \\ \psi\tsr df\end{pmatrix} 
= \begin{pmatrix}
    -(\laplace(a)/2) \psi\tsr w + \psi\tsr\inner{df, da} \\ \psi\tsr w\tsr da 
\end{pmatrix} 
\end{align}
Therefore, the inequality $\norm{\Theta(a)} \leq 4\dim(M)\norm{a}_{W^{2, \infty}}$ follows, since$$
\norm{\Theta(a)(\psi\tsr(w\oplus df))}^2 \leq  \norm{\laplace(a)/2}^2\norm{\psi \tsr w}^2 + \norm{da}^2\norm{\psi \tsr df}^2 + \norm{da}^2\norm{\psi\tsr w}^2  
$$
\end{proof}

Now $\Fc \subset W^{k, p}(M)$ since $\Fc\subset\C^\infty(M), M$ compact for all $p, k$. 
Because $\norm{\Theta(a)} \leq K\norm{a}_{W^{2, \infty}}$, if $J^{(n-1)}_t$ is bounded for each $t$, then $J_t^{(n)}$ is continuous on $\Fc$ wit respect to $W^{2, 2}$-norm topology. We will use this to show that if $J^{(n-1)}_t$ is bounded, then $j^{(n)}_t$ is positive on $\Kc$. Then using $j^{n}_t$ is positive on $\Kc$ it will be checked that for every $a\in \Fc$, $j^{n}_t(a)\in \Bc(\Kc)$ and that it extends from $\Bc(\Kc)$ to $\Bc(H\tsr \Fock(k))$. The base case is $j_t^{(0)} = \one \in \Bc(\Kc)$ which is obviously positive. Then from $j_t^n:\Fc\to \Bc(H\tsr \Fock(k))$, it extends to $j_t^n:\C(M)\to \Bc(H\tsr \Fock(k))$. 

\begin{lemma} Suppose $J^{(n-1)}_t\in \Bc(\Kc)$, then $j_t^n$ is a positive map on $a\in \Fc, a>0$. 
\end{lemma}
\begin{proof} Suppose $a\in \Fc$ is positive. We want to show $j_t^{n}(a)$ is positive as well. If $\sqrt{a}\in \Fc$, then \begin{align}\label{eq_positivity_on_squares}
\inner{u, j^n_t(a)u} = \inner{j^n_t(\sqrt{a})u, j^n_t(\sqrt{a})u}\geq 0
\end{align}
for every $u\in \Kc$, hence $j_t(a)$ is positive. 

So assume $\sqrt{a}\not\in\Fc$ where $a$ is positive, $a(m)>0, m\in M$. Since $\LinSpan(\Fc)$ is dense in $\C(M)$, for any $\eps>0$, there exists $f\in \Fc$ such that $\norm{\sqrt{a}-f}\leq O(\eps)$ meaning $\norm{a-f^2}\leq O(\eps)$.  Additionally $f$ can be chosen so $\norm{a-f^2}_{W^{2, \infty}} \leq O(\eps)$, so $f^2$ approximates $a$ in Sobolev $W^{2, \infty}$-norm as well. 

To see why this is possible note that since $a>0, \sqrt{a}\in C^\infty(M)$, therefore, $\sqrt{a}\in L^2(M)$, additionally for all $k$, $\grad^k(\sqrt{a})\in L^2(M)$, with $$\sqrt{a} = \sum_i \alpha_i \phi_i \txt{ and } \sum_i\lambda_i^{2k}\alpha_i =  \inner{\grad^k(\sqrt{a}),\grad^k(\sqrt{a})}_{L^2(M)} < \infty
$$ where $\alpha_i=\inner{\phi_i, \sqrt{a}}$ and Ricci-flatness was used. So the sequence $(\sum_{i=1}^n\alpha_i\phi_i )_{n\in \NN}$ converging to $\sqrt{a}$ in $L^2(M)$ is a bounded in every $W^{k, 2}$. For sufficiently large $k$, the embedding $W^{k, 2}(M) \subset W^{2, 2}(M)$ is compact by Rellich–Kondrachov theorem, that is, $(\sum_{i=1}^n\alpha_i\phi_i)_n$ has a Cauchy, and so a convergent subsequence; wlog let this subsequence be denoted by the same $ \sum_{i=1}^n\alpha_i\phi_i := a_n$. For $(a_i)$ to be convergent in $W^{2, 2}(M)$, $\norm{a_i}_{L^2(M)}$ must vanish, so the only possible limit is $\sqrt{a}$. Now suppose the tail $\sum_{i=n}^\infty\alpha_i\phi_i$ does not vanish in $W^{2, \infty}$. This means for some $x\in M$ 
for some $k \in \{0, 1, 2\}$, $\inner{\sum_i\alpha_i\grad^k\phi_i, \sum_i\alpha_i\grad^k \phi_i}(x) > 0$. But then by the following argument shows that $\inner{\sum_i\alpha_i\grad^k\phi_i, \sum_i\alpha_i\grad^k \phi_i}_{L^2(U)}> 0$
contradicting the convergence in $W^{2, 2}(M)$. So $\sqrt{a}$ can be approximated arbitrarily well in $W^{2, \infty}(M)$.

\begin{claim} 
Suppose $a=\sum_i \alpha_i\phi_i\in \C^\infty(M)$, then $\inner{\grad^k\sum_i \alpha_i\phi_i, \grad^k\sum_i \alpha_i\phi_i}(x) > 0$ for some $x\in M$ implies $\inner{\grad^k\sum_i \alpha_i\phi_i, \grad^k\sum_i \alpha_i\phi_i}_{L^2(M)} > 0$.
In particular, this holds for $k=0$. 
\end{claim}
\begin{proof} As before in proposition~\ref{prop_gradients_are_contractive} by smoothness of $a$, this holds for all $x\in U$ for some open set $U$, and with $\psi \geq 0$ compactly supported on $U$, $\psi>0$ on an open $V\subset U$, \begin{align*}
\sup_{x\in U}\psi(x)\cdot \inner{\grad^k a, \grad^k a}_{L^2(M)} &=  \sup_{x\in U}\psi(x)\cdot \int_M \inner{\grad^k a, \grad^k a}dV_g
\geq \int_U\inner{\psi\grad^k\sum_i \alpha_i \phi_i, \grad^k\sum_i \alpha_i\phi_i} dV_g > 0
\end{align*} 
Since $\psi$ is compactly supported and bounded, the claim follows. The $k=0, \grad^k=\one$ specialization is identical.
\end{proof}

\noindent Now define $$\Wc_a = \{a\}\cup \{f^2: f\in\Fc \text{ with } \norm{a-f^2}_{W^{2, \infty}}\leq 1/n, n\in \NN\}$$ then as $\norm{\Theta(a')}\leq K \norm{a'}_{W^{2, \infty}} $ and $J^{(n-1)}_t$ is bounded on $\Kc$ by hypothesis, the bound in lemma~\ref{lemma_picard__recursive_estimate}, implies norm is continuous map with respect to $\norm{\cdot}_{W^{2, \infty}}$-topology on $\Wc_a$: $$\norm{\cdot}:\Wc_a\to \RR, a'\to \norm{j_t(a')}$$ 

If $j^n_t(a)$ is not positive, then there exists $u\in \Kc$ such that $\inner{u, j^n_t(a)u} < 0$. Since norm is continuous, the map $a'\to \inner{u, j^n_t(a')u}$ is also continuous on $\Wc_a$: by Cauchy-Schwartz inequality, $\inner{u, j^n_t(a')u} \leq \norm{u}\norm{j^n_t(a')u} \leq \norm{u}^2K'K\norm{a'}_{W^{2, \infty}}$ where $K'$ depends on $u$ which we fixed and $\norm{J_t^{n-1}}$. This continuity means $\inner{u, j^n_t(\cdot)u} < 0$ on some neighborhood containing $a$ in $\Wc_a$. However, for any neighborhood $U$ of $a$ in $\Wc_a$ , $w\in U, w\neq a$ implies $w=f^2, f\in \Fc$, so $\inner{u, j^n_t(f^2)u} \geq 0$ by equation~\ref{eq_positivity_on_squares}. Therefore,  $j^n_t(a)$ must be positive. 
\end{proof}

\begin{lemma} If $j^{n}_t$ is a positive map on positive $a, a>0$, then $\norm{j^{n}_t(a)}\leq  \norm{a}^2$
\end{lemma}
\begin{proof} Let $x\in \Fc$ so $(1+\eps)\norm{x}-x\in \Fc$ and positive for any $\eps>0$. Define $\Phi_\eps(x):=\sqrt{(1+\eps)\norm{x}\one - x}\in \C(M)$.  Approximate $\Phi_\eps(x)$ from below by $z\in \Fc$. Then $j_t^{n}(\Phi_\eps(x)^2-z^2) > 0$ because $\Phi_\eps(x)^2-z^2>0$ and $j_t^{n}$ is positive. This yields $\inner{\theta, (j_t^{n}(\Phi_\eps(x)^2) - j_t^{n}(z^2))\theta} \geq 0$ and we have $$
\inner{\theta, j_t^{n}(\Phi_\eps(x)^2)\theta} \geq \inner{\theta, j_t^{n}(z^2)\theta} \geq 0
$$
Now the usual squareroot trick takes over: since 
 $j_t^{n}$ is unital, \begin{align}
0\leq \norm{j_t^{n}(z)\theta}^2 &= \inner{\theta, j_t^{n}(z^2)\theta} \leq \inner{\theta, j_t^{(n)}((1+\eps)\norm{x}\one - x)\theta} \\
\inner{\theta, j_t^{n}(x)\theta} &\leq \inner{\theta, j_t^{n}((1+\eps)\norm{x}\one)\theta}\leq (1+\eps)\norm{x}\inner{\theta, j_t^{n}(\one)\theta} = (1+\eps)\norm{x}\norm{\theta}^2 
\end{align}
Since $\eps$ was arbitrary, $\inner{\theta, j_t^{n}(x)\theta} \leq \norm{x}\norm{\theta}^2$. 
Finally, \begin{align}
\norm{j^{n}_t(x)\theta}^2 =\inner{j^{n}_t(x)\theta, j^{n}_t(x)\theta} = \inner{\theta, j^{n}_t(x^*x)\theta}\leq \norm{x^*x}\norm{\theta}^2 = \norm{x}^2\norm{\theta}^2
\end{align}.
So $\norm{j_t^n(x)}\leq \norm{x}^2$, and the bound on $\norm{j_t^n}$ is uniform.
\end{proof}

Now from density of $\Fc, \Vc$ and $\Kc$, each $j^n_t$ extends from a map $j_t^{n}:\Fc\to \Bc(\Kc)$ to $j_t^{n}:\Ac\to \Bc(H\tsr \Fock(k))$. Since $S_N(t)= \sum_{n\in [N]} J^{n}$ converges, so does $S=\lim_{N\to \infty} S_N$, and therefore $\lim_{n\to \infty}\sum j_t^{n}$ is the needed flow. Precisely, we have the following result: 

\begin{theorem} Following notation from theorem~\ref{app_thm_picard_iterates_on_F} from define $j_t(a)(v_1\ExpHVec f_1) := J_t(a\tsr \ExpHVec f_1)v_1$, then \begin{enumerate}
\item $j_t: \Fc\to \Bc(\Hs\tsr \ExpHVec(\Uc))$ is a unital $*$-homomorphism
\item $j_t$ extends to $j_t:\Fc\to \Bc(\Hs\tsr \Fock(k_0))$
\item $j_t$ extends to $j_t:\Ac\to \Bc(\Hs\tsr \Fock(k_0))$
\end{enumerate}
\end{theorem}

\begin{remark}{A remark on \namecite{goswami_qsp} Frechet structures:} Proposition~\ref{prop_sobolev_norm_bound_on_theta}, along with the growth bounds on $\ell(\grad^k\phi)$ suggests that convergence of the stochastic integrals can approached via a variant of complete smoothness regularity condition introduced by \cite{goswami_qsp}. In absence of the group action, the Frechet space structure on $k_0$ has to be obtained by other methods -- in this setting $\grad^k$ is the natural candidate for defining the Sobolev norms on  $d\Fc\subset k_0$. Some work is required to strap that into the Frechet machinery, since \cite{goswami_qsp} require norms to come from derivations on $\Ac=\C(M)$, but covariant derivative is a derivation on the tensor bundle. Note that the usual Sobolev norms $\norm{\cdot}_{W^{k, \infty}}$ can be used to obtain estimates of the type needed: for instance for estimate \ref{estimate_creation_annihilation_terms}, $i)$, using that $\norm{ab}_{W^{p, \infty}}\leq \norm{a}_{W^{p, \infty}}\norm{b}_{W^{p, \infty}}$ yields a slightly weaker bound of form $M_{f, x}^{n^2}$ which is still enough to get convergence. 
\end{remark}
\begin{remark}{A remark on \namecite{belton2015algebraic} construction:} The polynomial bounded growth condition that is leveraged for the Laplacian is similar to one obtained by \cite{belton2015algebraic}. However, this growth cannot be bound on the entire Hilbert space as operator $\Theta(a)$ is only defined on $\C^\infty(M)$, and it becomes necessary to work with a dense subspace and utilize some mild continuity to push the necessary estimates through. The continuity is also needed for the extension of the squareroot trick.     
\end{remark}

\hsep

\appendix
\section{Laplacian on products}\label{section_laplacian_products}

\begin{prop}[Laplacian on products] For $\phi_j,\phi_i$,
\begin{enumerate}[label=\roman*]
\item $\laplace(\phi_i\phi_j) = \phi_i\laplace(\phi_j) + \phi_i\laplace(\phi_j) + 2\inner{\grad \phi_i, \grad \phi_j}$
\item For any $k$, \begin{align*}\laplace\inner{\grad^k\phi_i, \grad^k\phi_j} = \inner{\grad^k\laplace\phi_i, \grad^k\phi_j} + \inner{\grad^k\phi_i, \grad^k\laplace\phi_j} + 2\inner{\grad^{k+1}\phi_i, \grad^{k+1}\phi_j}\end{align*}
\end{enumerate}
\end{prop}
\begin{proof} With $\laplace=g^{ab}\grad_a\grad_b$, $M$ Ricci flat, using $g^{ij}=g^{ji}, \grad_a(\phi_i\phi_j)=\phi_j\grad_a(\phi_i) + \grad_a(\phi_j)\phi_i$ yields: \begin{align*} \laplace(\phi_i\phi_j) &= g^{ab}\grad_a\grad_b(\phi_i\phi_j) = g^{ab}(\grad_a\grad_b(\phi_i\phi_j)) =  g^{ab}(\grad_a(\phi_i\grad_b(\phi_j)+ \phi_j\grad_b(\phi_i)))\\
&= f\laplace(\phi_j) + \phi_j\laplace(\phi_i) + 2g^{ab}\grad_a(\phi_i)\grad_b(\phi_j) = f\laplace(\phi_j) + \phi_j\laplace(\phi_i) + 2\inner{\grad \phi_i, \grad \phi_j}
\end{align*}
For $\laplace = -g^{ab}\grad_a\grad_b$, we get $-2\inner{\grad \phi_i, \grad \phi_j}$ instead for the mixed term. 

\noindent To see $ii)$, note that by $i)$ and for $k$-tensors $f, h$ with multi-indices $a, b$, $g^{ab}:=g^{a_1b_1}\dots g^{a_kb_k}$, we have \begin{align*}\laplace\inner{f, h} = \laplace g^{ab}f_ah_b &= g^{ab}\laplace f_ah_b = g^{ab}f_a \laplace h_b  + g^{ab}(\laplace f_a)h_b + 2g^{ab}g^{cd}\grad_c f_a \grad_d h_b\\ 
&= \inner{\laplace f, h} + \inner{f, \laplace h} + 2 \inner{\grad f, \grad h}\end{align*} 
and then since$[\grad, \laplace]=0$ as $M$ is Ricci-flat, as needed 
\begin{align*}\laplace\inner{\grad^k\phi_i, \grad^k\phi_j} = \inner{\grad^k\laplace\phi_i, \grad^k\phi_j}
  + \inner{\grad^k\phi_i, \grad^k\laplace\phi_j} + \inner{\grad^{k+1}\phi_i, \grad^{k+1}\phi_j}
\end{align*}
 
\end{proof}


\section{Appendix: Map-valued Evans-Hudson quantum sde's}\label{section_mapvalued_qsdes}

Suppose $\Ac_0$ is a  dense $*$-subalgebra inside the \CStar-algebra $\Ac\subset \Bc(\Hs)$, $k_0$ is the noise space with $k:=L^2(\RR_+, k_0), k_t:=L^2((0, t), k_0), k^t:=L^2((t, \infty), k_0)$. $\Fock=\Fock_s(k)$. $\Fock_t := \Fock(L^2((0, t), k_0), \Fock^t := \Fock(L^2((t, \infty), k_0)$. Similarly, $f\in k, f^t, f_t$ are projections onto $k^t, k_t, f_t:=f\one_{[0, t]}, f^t:f\one_{[t, \infty)}$. $\hat k_0 := \CC\oplus K_0$. 

We will only work with annihilation and creation processes since the conservation process in the Laplacian generated processes is zero. For a map $\txt{Lin}(D_0, \Hs\tsr k_0)$, $D_0\DenseIn \Hs$, if $R(u):=a\tsr b$, then  for $\Delta\subset (t, \infty)$, using the mapping $k_0\ni b\to b\one_\Delta\in k_t$, $R^\Delta_t$ is defined by $$
D_0\tsr \Fock_t \ni u\tsr \psi \to R^\Delta_t(u\tsr \psi) := a\tsr \psi \tsr (b\one_{\Delta}) \in \Hs\tsr \Fock_t\tsr k^t 
$$
The associated creation process creates $k_0$ component of $R$ on interval $\Delta$: $$
a_R^\dagger(\Delta):=a^\dagger(R_t^\Delta)$$
The corresponding annihilation process is defined by using $k_0$ component of $R$ to annihilate: $$
(D_0\tsr \Fock_t)\tsr \Fock^t \ni u_t\ExpHVec(f^t) \to a_R(\delta)(u_t\ExpHVec(f^t)) = \left(\left(\int_\Delta \inner{R, f(s)}ds\right)u_t\right)\ExpHVec(f^t)
$$
where $\inner{R, f(s)}$ is viewed as an operator(\FRem{\cite[pg~93]{goswami_qsp}}) on $\Hs$.    

The Hudson-Parthasarathy quantum stochastic calculus on Hilbert spaces\FRem{\cite[\S~5.2]{goswami_qsp}} are set in the
the Schr\"{o}edinger picture of quantum dynamics. Suppose $(H_t)_{t\geq 0}$ is a family of linear operators on $\Hs\tsr \Fock$ with $\{v f_t^{\tsr n} \psi^t \}\subset \Dom(H_t)$ for $v\in D_1\DenseIn \Hs, f_t\in k_t$, $f_t$ simple, right continuous and valued in $V\DenseIn k_0, \psi^t\in \Fock^t$, with $H_t = \hat H_t \tsr \one_{\Fock^t}$ for some map $\hat H_t: \{\Hs\tsr \ExpHVec(k_t)\}\supset \Dom(\hat H_t)\to \Hs\tsr \Fock$, such that $\sup_{-\leq s\leq t}\norm{H_s(u\ExpHVec(f))}\leq \norm{r_tu}$ for all $t$, where $r_t$ depending on $t,f$ only is a closable map in $ \Lin(D_1, H')$ for some Hilbert space $H'$ depending only on $f$. 


If $H_t$ is simple, that is, $H_t=\sum_{i=0}^mH_{t_i}\one_{[t_i, t_{i+1})}(t)$, $0=t_0<\dots t_m < t_{m+1}=\infty$, then for $M$ as one of the fundamental processes\footnote{The integral with respect to the conservation process is not treated here, but the treatment is analogous} $a_R, a^\dagger_R, t\one$, $$
\int_0^t H_sM(ds) = \sum_0^m H_{t_i}M([t_i, t_{i+1})\cap [0, t])
$$ 
For general processes the integral is obtained as an appropriate limit. 

The Heisenberg formalism is captured by operator valued processes: for an adapted, regular process, $Y(t):\Ac\tsr \Fock \supset \Dom(Y(t))\to \Ac\tsr \Fock$, $\wtilde{Y(t)}: \Ac\tsr k_0 \tsr \Fock_t \supset \Dom(\wtilde{Y(t)})\to \Ac\tsr \Fock_t\tsr k_0$ defined by $\wtilde{Y(t)} = (Y(t)\tsr \one_{k_0})\SwapOp{23}$ where $\SwapOp{23}(a_1\tsr a_2 \tsr a_3) = a_1\tsr a_3\tsr a_2$, $$
\wtilde{Y(s)}:\Ac\tsr k_0\tsr \ExpHVec(k_s) \xrightarrow{\SwapOp{23}\txt{ }} \Ac\tsr  \ExpHVec(k_s)\tsr k_0 \xrightarrow{Y(s)\tsr \one} \Ac\tsr  \Fock(k_s)\tsr k_0  
$$
The map-valued processes for any $\delta:\Ac_0\to \Ac\tsr k_0$ are defined as follows\FRem{\cite[pg~134]{goswami_qsp}} \begin{align}
a_\delta(\Delta)(\sum_i x_i\tsr \ExpHVec(f_i))u &:= \sum_i a_{\delta(x_i^*)}(\Delta)(u\ExpHVec(f_i)),\ \ a^\dagger_\delta(\Delta)(\sum_i x_i\tsr \ExpHVec(f_i))u := \sum_i a_{\delta}^\dagger(x_i)(\Delta)(u\ExpHVec(f_i))\label{eq_fundamental_processes1} \\
I_\Lc(\Delta)(\sum_i x_i\tsr \ExpHVec(f_i))u &:= \sum_i |\Delta|(\Lc(x_i)u)\tsr \ExpHVec(f_i)\label{eq_fundamental_processes2}
\end{align}
and the map-valued integrals are defined by\FRem{\cite[136,137]{goswami_qsp}}: with $u\in \Hs, f\in L^4_{\txt{loc}}, x\in \Ac$,\begin{align}\label{eqs_deinition_mapvalued_integrals}
\left(\int_0^t Y(s)\circ (a_\delta + I_\Lc)(ds)   \right)(x\tsr \ExpHVec(f))u &= \int_0^tY(s)\left((\Lc(x) + \inner{\delta(x^*), f(s)}\tsr \ExpHVec(f)\right)u\ ds \\ 
\left(\int_0^t Y(s)\circ (a^\dagger_\delta)(ds)   \right)(x\tsr \ExpHVec(f))u &= \left(\int_0^t a^\dagger_{\wtilde{Y}, x}(ds)\right)u\ExpHVec(f)\label{eq_ytilde_integrator_in_terms_of_creation}
\end{align}
where \begin{align}\label{eq_ytilde_integrator}
a^\dagger_{\wtilde{Y}, x}(s)(u\ExpHVec(f_s)):=\wtilde{Y(s)}(\delta(x)\tsr \ExpHVec(f_s))u
\end{align}
\begin{remark}\label{remark_extending_ytilde} A standard technique involves using iterated integrals arising through Picard iteration, requiring iteration of the map $\delta:\Ac\to k_0$ of the form $\delta(\delta(x)) := (\delta\tsr \one)(\delta (x)):\Ac\to \Ac\tsr k_0\tsr k_0$, more generally, $\delta^m:\Ac\to \Ac\tsr k_0^{\tsr m}$. To this, $\wtilde Y$ is extended as follows -- $$
\wtilde{Y(s)}:\Ac\tsr \left(\oplus_m k_0^{\tsr m}\right)\tsr \ExpHVec(k_s) \xrightarrow{\SwapOp{23}\txt{ }} \Ac\tsr  \ExpHVec(k_s)\tsr \left(\oplus_m k_0^{\tsr m}\right) \xrightarrow{Y(s)\tsr \one} \Ac\tsr  \Fock(k_s)\tsr \left(\oplus_m k_0^{\tsr m}\right)
$$
To make sense of the creation process on $k_0^{\tsr m}$, one just needs an identification $k_0^{\tsr m}\to k_t$, a natural choice is via $a_1\tsr a_2\dots a_m \to \sum_{i\in [0:m-1]}a_i\one_{i|\Delta|+\Delta}$, where the tensor product leads to the interpretation as a partition refinement.

\end{remark}

Now assume $\Lc(\Ac^\infty)\subset \Ac_\infty$. 
With $\Vc_0=(k_0)_\infty$, define \begin{itemize}
\item $\Vc_t=\{\Vc_0\txt{-valued simple functions in }k_t\}$
\item $\Vc=\{\Vc_0\txt{-valued simple functions}\}$
\end{itemize}
One defines a map-valued integrable process with respect to $a_\delta, a^\dagger_\delta, I_\Lc$ as follows: 
\begin{definition}\label{defn_integrabile_process}
An integrable completely smooth map-valued process is an adapted process $(Y(s))_{s\geq 0}:\Ac_\infty\tsr \ExpHVec(\Vc)\to \Ac\tsr \Fock(k)$ such that:\begin{enumerate}
\item For each $t\geq 0, f\in \Vc$, $Y(t)(a\tsr\ExpHVec(f))\in (\Ac\tsr \Fock(k))_\infty$ and the map following map is completely smooth:\begin{align}
\Ac_\infty\ni a\to \Omega_{t, f}(a):Y(t)(a\tsr \ExpHVec(f))\in (\Ac\tsr \Fock(k))_\infty
\end{align}  
\item For every $m\in \NN$ and fixed $X\in \Ec^m_\infty:=(\Ac\tsr \hat k_0^{\tsr m})_\infty$, $f\in \Vc$, any separable Hilbert space $\Hs'$, the ampilation $\wtilde \Omega_{t, f}:=\Omega_{t, f}\tsr \one_{\Hs'}:(\Ac\tsr \Hs')_\infty\to (\Ac\tsr \Fock\tsr \Hs')_\infty$ is continuous.
\item For every fixed $a\in \Ac_\infty, f\in \Vc$, $\wtilde{Y(s)}:=\wtilde \Omega_{t, f}$, \begin{align}
S_a(s):\Hs\tsr \ExpHVec(\Vc_s)\ni u\ExpHVec(f_s)\to \wtilde{Y(s)}(\delta(a\tsr \ExpHVec(f_s)))u \in \Hs\tsr \Fock_s\tsr k_0
\end{align} is continuous
\item For every $a\in \Ac_\infty, \xi\in \Vc_0$ the map \begin{align} s\to Y(s)((\Lc(a)) +\inner{\delta(a^*), \xi}))\tsr \ExpHVec(f)
\end{align} is strongly continuous. 
\end{enumerate}
\end{definition}

\printbibliography[title={\sectiontitlefont{References}}]

\end{document}